%% file: paper.tex
\documentclass[a4paper,UKenglish,cleveref, autoref, thm-restate]{lipics-v2021}

\pdfoutput=1 
\hideLIPIcs  

\usepackage[inline]{enumitem}
\usepackage{tikz}
\usetikzlibrary{arrows,positioning,decorations.markings}
\tikzset{
  textnode/.style={},
  varnode2/.style={draw, outer sep=1pt},
  varnode2green/.style={draw,fill=green},
  varnode2red/.style={draw,fill=red},
  varnode2blue/.style={draw,fill=blue},
  varnode/.style={rectangle},
  varnodegreen/.style={rectangle,draw,fill=green},
  varnodered/.style={rectangle,draw,fill=red},
  varnodeblue/.style={rectangle,draw,fill=blue},
  varnodepurple/.style={rectangle,draw,fill=purple},
  varnodeyellow/.style={rectangle,draw,fill=yellow},
  varnodeempty/.style={inner sep=0pt,fill},
  line/.style={=stealth,thick,fill=red},
  matchededge/.style={>=stealth,thick,fill=black},
  unmatchededge/.style={>=stealth,thick,fill=blue},
  arrow/.style={=>stealth,thick}
}
\usepackage{amsthm}
\usepackage{amssymb}
\usepackage[ruled,noline,noalgohanging]{algorithm2e}

\newtheorem{mylem}{Lemma}

\newtheorem{mythm}{Theorem}
\newtheorem{myprop}{Proposition}

\usepackage{listings}
\lstdefinestyle{inline}{%
    basicstyle=\ttfamily\small%
}

\usepackage{mathtools}

\input{pddl-listings}

\title{A Formal Analysis of RANKING} 

\author{Mohammad Abdulaziz}{King's College London, United Kingdom\and Technische Universität München, Germany}{mansour@in.tum.de}{https://orcid.org/0000-0002-8244-518X}{Part of this work was funded by DFG Koselleck Grant
  NI 491/16-1}

\author{Christoph Madlener}{Technische Universität München, Germany}{madlener@in.tum.de}{https://orcid.org/0000-0002-9577-0061}{}

\authorrunning{M. Abdulaziz and C. Madlener} 

\Copyright{Mohammad Abdulaziz and Christopher Madlener} 



\category{} 

\relatedversion{} 

\nolinenumbers 

\EventEditors{John Q. Open and Joan R. Access}
\EventNoEds{2}
\EventLongTitle{42nd Conference on Very Important Topics (CVIT 2016)}
\EventShortTitle{CVIT 2016}
\EventAcronym{CVIT}
\EventYear{2016}
\EventDate{December 24--27, 2016}
\EventLocation{Little Whinging, United Kingdom}
\EventLogo{}
\SeriesVolume{42}
\ArticleNo{23}

\begin{document}
\ifpdf
\else
\begin{center}
\textbf{Warning:} The DVI version of this paper may be corrupted.  If possible, use the PDF version.
\end{center}
\fi
\renewcommand{\todo}[1]{}
\maketitle

\begin{abstract}
We describe a formal correctness proof of RANKING, an online algorithm for online bipartite matching.
An outcome of our formalisation is that it shows that there is a gap in all combinatorial proofs of the algorithm.
Filling that gap constituted the majority of the effort which went into this work.
This is despite the algorithm being one of the most studied algorithms and a central result in theoretical computer science.
This gap is an example of difficulties in formalising graphical arguments which are ubiquitous in the theory of computing.
\end{abstract}
\input{macros}
\input{isabelle}
\input{intro}

\input{definitions}

\input{probabilistic}
\input{combinatorial}
\input{conclusion}
\bibliography{long_paper}
\pagebreak
\appendix
\input{supp}

\end{document}

%% file: pddl-listings.tex
\lstdefinestyle{inline}{%
    basicstyle=\ttfamily\small%
}

\usepackage[obeyFinal]{todonotes}
\usepackage{bm}

\usepackage[normalem]{ulem}
\usepackage{lineno}

\usepackage[most]{tcolorbox}
\tcbuselibrary{breakable,hooks,skins,theorems}
\newtcbtheorem{IsabelleTheorem}{Listing}{
    enhanced jigsaw
    , breakable
    , float=h
    , top=0pt
    , right=0pt
    , bottom=0pt
    , left=0pt
    , boxrule=0pt
    , toprule=1pt
    , bottomrule=1pt
    , titlerule=1pt
    , colframe=black
    , sharp cornersX
    , colbacktitle=white
    , coltitle=black
    , fonttitle=\tiny
    , colback=white
    , fontupper=\ttfamily\tiny
    , before upper={\parindent0em}
}{isa}

\lstdefinelanguage{isabelle}{
    morekeywords={record,type_synonym,definition,fun,function,primrec,where,lemma,theorem,unfolding,by,shows,assumes,and,datatype,using,abbreviation
,moreover,have,hence,thus,qed,proof,let,ultimately,show,next,in}
    , sensitive=true
    , showstringspaces=false
    , framerule=0pt
    , xleftmargin=2em
    , numbers=left
    , numberstyle=\ttfamily\tiny
    , firstnumber=1
    , stepnumber=2
    , basicstyle=\sffamily\scriptsize
    ,backgroundcolor = \color{white}
    , breaklines=true
    , showspaces=false
    , morecomment=[l]{--}
    , morecomment=[s]{(*}{*)}
    , commentstyle=\color{gray}
    , morestring=[b]"
    , literate={\\<times>}{{$\times$}}{1} {\\<equiv>}{{$\equiv$}}{1} {\\<forall>}{{$\forall$}}{1} {\\<exists>}{{$\exists$}}{1} {\\<and>}{{$\land$}}{1}
        {\\<in>}{{$\in$}}{1} {\\<Rightarrow>}{{$\Rightarrow$}}{1} {\\<lambda>}{{$\lambda$}}{1} {::}{{$::$}}{1}
        {\\<subseteq>}{{$\subseteq$}}{1} {\\<^sub>m}{{$_m$}}{1} {\\<longleftrightarrow>}{{$\longleftrightarrow$}}{3}
        {\\<pi>}{{$\pi$}}{1} {\\<delta>}{{$\delta$}}{1} {\\<lbrakk>}{{$\llbracket$}}{1} {\\<rbrakk>}{{$\rrbracket$}}{1}
        {\\<Longrightarrow>}{{$\Longrightarrow$}}{3} {\\<not>}{{$\lnot$}}{1} {\\<le>}{{$\le$}}{1} {\\<rightharpoonup>}{{$\rightharpoonup$}}{2}
        {\\<^sub>\\<V>}{{$_{\mathcal V}$}}{1} {\\<lparr>}{{$\llparenthesis$}}{1} {\\<rparr>}{{$\rrparenthesis$}}{1}
        {\\<leftarrow>}{{$\leftarrow$}}{1} {\\<^sub>\\<O>}{{$_{\mathcal O}$}}{1} {\\<^sub>I}{{$_{\texttt{I}}$}}{1}
        {\\<^sub>G}{{$_{\texttt{G}}$}}{2} {\\<phi>}{{$\varphi$}}{1} {\\<Phi>}{{$\Phi$}}{1} {\\<psi>}{{$\psi$}}{1} {\\<Psi>}{{$\Psi$}}{1}
        {\\<^sub>S}{{$_{\texttt S}$}}{1} {\\<inverse>}{{$^{-1}$}}{1} {\\<^sub>O}{{$_{\texttt O}$}}{1} {\\<^bold>\\<And>}{{$\bm\bigwedge$}}{1}
        {\\<^bold>\\<or>}{{$\bm\lor$}}{1} {\\<^sub>G}{{$_{\texttt G}$}}{1} {\\<Pi>}{{$\Pi$}}{1} {\\<^sub>I}{{$_{\texttt I}$}}{1} {\\<noteq>}{{$\neq$}}{1}
        {\\<bottom>}{{$\bot$}}{1} {\\<^sub>+}{{$_\texttt +$}}{1} {\\<^bold>\\<and>}{{$\bm\land$}}{1} {\\<^bold>\\<not>}{{$\bm\lnot$}}{1}
        {\\<^sub>1}{{$_1$}}{1} {\\<^sub>2}{{$_2$}}{1} {\\<A>}{{$\mathcal A$}}{1} {\\<Turnstile>}{{$\models$}}{2} {\\<^sub>\\<forall>}{{$_\forall$}}{1}
        {\\<^sub>0}{{$_0$}}{1} {\\<tau>}{{$\tau$}}{1}  {\\<^sub>\\<Omega>}{{$_\Omega$}}{1} {\\<^sub>V}{{$_V$}}{1} {\\<^bold>\\<Or>}{{$\bm\bigvee$}}{1}
        {\\<^sub>P}{{$_\texttt P$}}{1} {\\<^sub>X}{{$_\texttt X$}}{1} {\\<longrightarrow>}{{$\longrightarrow$}}{2} {\\<or>}{{$\lor$}}{1} {\\<^sub>\\<pi>}{{$_\pi$}}{1}
        {\\<^sub>s}{{$_s$}}{1} {\\<^sub>t}{{$_t$}}{1} {\\<^sub>a}{{$_a$}}{1} {\\<^sub>r}{{$_r$}}{1} {\\<^sub>t}{{$_t$}}{1} {\\<^sub>e}{{$_e$}}{1} {\\<^sub>n}{{$_n$}}{1} {\\<^sub>d}{{$_d$}}{1} {\\<^sub>i}{{$_i$}}{1} {\\<^sub>v}{{$_v$}}{1} {\\<^sub>j}{{$_j$}}{1} {\\<^sub>b}{{$_b$}}{1} {\\<inter>}{{$\cap$}}{1} {\\<union>}{{$\cup$}}{1} {\\<Union>}{{$\bigcup$}}{1} {\\<^sup>c\\<TTurnstile>\\<^sub>=}{{${}^c\models_=$}}{1}
        {\\<open>}{{<}}{1} {\\<close>}{{>}}{1} {\\<langle>}{{$\langle$}}{1} {\\<rangle>}{{$\rangle$}}{1} {\\<ge>}{{$\ge$}}{1} {\\<^sup>-\\<^sup>1\\<^sub>C}{{$^{\texttt{-1}}_\texttt{C}$}}{2} {\\<^sup>+\\<^sub>C}{{$^\texttt{+}_\texttt{C}$}}{1} {\\<circ>\\<^sub>C}{{$\circ^\texttt C$}}{1} {\\<top>\\<^sub>C}{{$\top_\texttt{C}$}}{2} {\\<bottom>\\<^sub>C}{{$\bot_\texttt{C}$}}{2} {\\<not>\\<^sub>C}{{$\neg_\texttt{C}$}}{2}
        {\\<squnion>\\<^sub>C}{{$\sqcup_\texttt{C}$}}{2} {\\<sqinter>\\<^sub>C}{{$\sqcap_\texttt{C}$}}{2} {\\<exists>\\<^sub>C}{{$\exists_\texttt{C}$}}{2}  {\\<forall>\\<^sub>C}{{$\forall_\texttt{C}$}}{2} {=\\<^sub>C}{{$=_\texttt{C}$}}{2} {\\<sigma>}{{$\sigma$}}{2} {\\<notin>}{{$\notin$}}{1} {\\<oplus>}{{$\oplus$}}{1} {\\<nexists>}{{$\nexists$}}{1} {\\<setminus>}{{$\setminus$}}{1} {_}{{$-$}}{1}
}
\newtcblisting[auto counter]{IsabelleSnippet}[2][]{
    listing options={
        language=isabelle
    }
    , listing only
    , enhanced jigsaw
    , breakable
    , top=-.5em
    , right=0pt
    , bottom=-.5em
    , left=0pt
    , boxrule=0pt
    , toprule=1pt
    , bottomrule=1pt
    , titlerule=1pt
    , colframe=black
    , sharp corners
    , colbacktitle=white
    , coltitle=black
    , colback=white
    , fontupper=\ttfamily\tiny
    , before upper={\parindent0em}
    , title=Listing \thetcbcounter: #2
    , #1
}

\lstdefinelanguage{pddl}{
  sensitive=false,    
  morecomment=[l]{;}, 
  alsoletter={:,-},   
  morekeywords={
    define,domain,problem,not,and,or,when,forall,exists,either,
    :domain,:requirements,:types,:objects,:constants,
    :predicates,:action,:parameters,:precondition,:effect,
    :fluents,:primary-effect,:side-effect,:init,:goal,
    :strips,:adl,:equality,:typing,:conditional-effects,
    :negative-preconditions,:disjunctive-preconditions,
    :existential-preconditions,:universal-preconditions,:quantified-preconditions,
    :functions,assign,increase,decrease,scale-up,scale-down,
    :metric,minimize,maximize,
    :durative-actions,:duration-inequalities,:continuous-effects,
    :durative-action,:duration,:condition
  },
  numberstyle=\ttfamily\tiny,
  basicstyle=\ttfamily\tiny
}

\newtcblisting[use counter from=IsabelleSnippet]{PddlListing}[2][]{
  listing options={
    language=pddl,
    escapechar=|
  },
  listing only,
  enhanced jigsaw,
  breakable,
  top=-.5em,
  right=0pt,
  bottom=-.5em,
  left=0pt,
  boxrule=0pt,
  toprule=1pt,
  bottomrule=1pt,
  titlerule=1pt,
  colframe=black,
  sharp corners,
  colbacktitle=white,
  coltitle=black,
  fonttitle=\tiny,
  colback=white,
  fontupper=\ttfamily\tiny,
  before upper={\parindent0em},
  title=Listing \thetcbcounter: #2,
  #1
}

\lstdefinestyle{isainline}{
  language=isabelle,
  basicstyle=%
    \ttfamily\small
}

%% file: macros.tex
\providecommand{\insts}{}
\renewcommand{\insts}{\ensuremath{\Delta}}
\providecommand{\inst}{\ensuremath{\tvsal}}
\newcommand{\act}{\ensuremath{\pi}}
\newcommand{\asarrow}[1]{\vec{#1}}
\renewcommand{\vec}[1]{\overset{\rightarrow}{#1}}
\newcommand{\as}{\ensuremath{\vec{{\act}}}}

\newcommand{\etc}{\textit{etc.}}
\newcommand{\versus}{\textit{vs.}}

\newcommand{\Ie}{I.e.}
\newcommand{\eg}{e.g.}
\newcommand*{\ie}{i.e.\ }
\newcommand{\abziz}[1]{\textcolor{brown}{#1}}
\newcommand{\sublist}[2]{ \ensuremath{#1} \preceq\!\!\!\raisebox{.4mm}{\ensuremath{\cdot}}\; \ensuremath{#2}}
\newcommand{\subscriptsublist}[2]{\ensuremath{#1}\preceq\!\raisebox{.05mm}{\ensuremath{\cdot}}\ensuremath{#2}}
\newcommand{\PLS}{\Pi^\preceq\!\raisebox{1mm}{\ensuremath{\cdot}}}
\newcommand{\PLScharles}{\Pi^d}
\newcommand{\execname}{\mathsf{ex}}
\newcommand{\IndHyp}{\mathsf{IH}}
\newcommand{\exec}[2]{#2(#1)}

\newcommand{\ancestorssymbol}{\textsf{\upshape ancestors}}
\newcommand{\ancestors}{\ancestorssymbol}
\newcommand{\satpreas}[2]{\ensuremath{sat_precond_as(s, \as)}}
\newcommand{\proj}[2]{\ensuremath{#1{\downharpoonright}_{#2}}}
\newcommand{\dep}[3]{\ensuremath{#2 {\rightarrow} #3}}
\newcommand{\deptc}[3]{\ensuremath{#2 {\rightarrow^+} #3}}
\newcommand{\negdep}[3]{\ensuremath{#2 \not\rightarrow #3}}
\newcommand{\leavessymbol}{\textsf{\upshape leaves}}
\newcommand{\leaves}{\leavessymbol}

\newcommand{\childrensymbol}{\textsf{\upshape children}}
\newcommand{\children}[2]{\mathcal{\childrensymbol}_{#2}(#1)}
\newcommand{\succsymbol}{\textsf{\upshape succ}}
\newcommand{\succstates}[2]{\succsymbol(#1, #2)}
\newcommand{\concat}{\#}
\newcommand{\RG}{\cite{Rintanen:Gretton:2013}\ }
\newcommand{\cupdot}{\charfusion[\mathbin]{\cup}{\cdot}}
\newcommand{\bigcupdot}{\charfusion[\mathop]{\bigcup}{\cdot}}
\newcommand{\cuparrow}{\charfusion[\mathbin]{\cup}{{\raisebox{.5ex} {\smathcalebox{.4}{\ensuremath{\leftarrow}}}}}}
\newcommand{\bigcuparrow}{\charfusion[\mathop]{\bigcup}{\leftarrow}}
\newcommand{\finiteunion}{\cuparrow}
\newcommand{\finitemap}{\ensuremath{\sqsubseteq}}
\newcommand{\dgraph}{dependency graph}
\newcommand{\domain}[1]{{\sc #1}}
\newcommand{\solver}[1]{{\sc #1}}
\providecommand{\problem}[1]{\domain{#1}}
\renewcommand{\v}{\ensuremath{\mathit{v}}}
\providecommand{\vs}[1]{\domain{#1}}
\renewcommand{\vs}{\ensuremath{\mathit{vs}}}
\newcommand{\VS}{\ensuremath{\mathit{VS}}}
\newcommand{\Aut}{\ensuremath{\mathit{Aut}}}
\newcommand{\Inst}[2]{\ensuremath{\mathit{#2 \rightarrow_{#1} #1}}}
\newcommand{\Image}{\ensuremath{\mathit{Im}}}
\newcommand{\Img}[2]{\protect{#1 \llparenthesis #2 \rrparenthesis}}
\newcommand{\SND}{\ensuremath{\mathit{\pi_2}}}
\newcommand{\FST}{\ensuremath{\mathit{\pi_1}}}
\newcommand{\tvsal}{{\pitchfork}}
\newcommand{\nauty}{CGIP}

\newcommand{\pwinter}{\ensuremath{\mathit{\bigcap_{pw}}}}

\newcommand{\dom}{\ensuremath{\mathit{\mathcal{D}}}}
\newcommand{\codom}{\ensuremath{\mathcal{R}}}

\newcommand{\map}{\ensuremath{\mathit{map}}}
\newcommand{\BIJEC}{\ensuremath{\mathit{bij}}}
\newcommand{\INJ}{\ensuremath{\mathit{inj}}}
\newcommand{\funion}{\ensuremath{\overset{\leftarrow}{\cup}}}

\newcommand{\ifnew}{\mbox{\upshape \textsf{if}}}
\newcommand{\thennew}{\mbox{\upshape \textsf{then}}}
\newcommand{\elsenew}{\mbox{\upshape \textsf{else}}}
\newcommand{\choice}{\mbox{\upshape \textsf{ch}}}
\newcommand{\arbchoice}{\mbox{\upshape \textsf{arb}}}
\newcommand{\acycchoice}{\mbox{\upshape \textsf{ac}}}
\newcommand{\cycchoice}{\mbox{\upshape \textsf{cyc}}}
\newcommand{\filter}{\ensuremath{\mathit{FIL}}}
\newcommand{\probset}{\ensuremath{\boldsymbol \Pi}}
\newcommand{\probleq}{\ensuremath{\leq_\Pi}}
\newcommand{\CommVar}{\ensuremath{\bigcap_\v} }
\newcommand{\quotfun}{\ensuremath{ \mathcal{Q}}}

\newcommand{\apre}{\mbox{\upshape \textsf{pre}}}
\newcommand{\aeff}{\mbox{\upshape \textsf{eff}}}
\newcommand{\problist}{\ensuremath \probset}
\newcommand{\cat}{{\frown}}
\newcommand{\probproj}[2]{{#1}{\downharpoonright}^{#2}}
\newcommand{\preced}{\mathbin{\rotatebox[origin=c]{180}{\ensuremath{\rhd}}}}
\newcommand{\perm}{\ensuremath{\sigma}}
\newcommand{\invariant}[2]{\ensuremath{\mathit{inv({#1},{#2})}}}
\newcommand{\invstates}[1]{\ensuremath{\mathit{inv({#1})}}}
\newcommand{\probss}[1]{{\mathcal S}(#1)}
\newcommand{\parChildRel}[3]{\ensuremath{\negdep{#1}{#2}{#3}}}
\newcommand{\asessymbol}{\ensuremath{\mathbb{A}}}
\newcommand{\ases}[1]{{#1}^*}
\newcommand{\uniStates}{\ensuremath{\mathbb{U}}}
\newcommand{\recurrenceDiam}{\ensuremath{\mathit{rd}}}
\newcommand{\recurrenceAcycDiamfun}{\ensuremath{\mathit{{\mathfrak A}}}}
\newcommand{\recurrenceDiamfun}{\ensuremath{\mathit{\mathfrak R}}}
\newcommand{\traversalDiam}{\ensuremath{\mathit{td}}}
\newcommand{\traversalDiamfun}{\ensuremath{\mathit{\mathfrak T}}}
\newcommand{\isPrefix}[2]{\ensuremath{#1 \preceq #2}}
\providecommand{\path}{\ensuremath{\gamma}}
\newcommand{\aspath}{\ensuremath{\vec{\path}}}
\renewcommand{\path}{\ensuremath{\gamma}}
\newcommand{\n}{\textsf{\upshape n}}
\providecommand{\graph}{}
\renewcommand{\graph}{{\cal G}}
\newcommand{\undirgraph}{{\cal G}}

\renewcommand{\ss}{\ensuremath{\state s}}
\newcommand{\slist}{\ensuremath{\vec{\mbox{\upshape \textsf{ss}}}}}
\newcommand{\sll}{\ensuremath{\vec{\state}}}
\newcommand{\listset}{\mbox{\upshape \textsf{set}}}
\newcommand{\asset}{\ensuremath{\mathit{K}}}
\newcommand{\aslist}{\ensuremath{\mathit{\overset{\rightarrow}{\gamma}}}}
\newcommand{\head}{\mbox{\upshape \textsf{first}}}
\renewcommand{\max}{\textsf{\upshape max}}
\newcommand{\argmax}{\textsf{\upshape argmax}}
\newcommand{\argmin}{\textsf{\upshape argmin}}
\renewcommand{\min}{\textsf{\upshape min}}
\newcommand{\bool}{\mbox{\upshape \textsf{bool}}}
\newcommand{\last}{\mbox{\upshape \textsf{last}}}
\newcommand{\front}{\mbox{\upshape \textsf{front}}}
\newcommand{\rot}{\mbox{\upshape \textsf{rot}}}
\newcommand{\stuff}{\mbox{\upshape \textsf{intlv}}}
\newcommand{\tail}{\mbox{\upshape \textsf{tail}}}
\newcommand{\ngrtoas}{\ensuremath{\mathit{\as_{\graph_\mathbb{N}}}}}
\newcommand{\vsfun}{\mbox{\upshape \textsf{vs}}}
\newcommand{\inits}{\mbox{\upshape \textsf{init}}}
\newcommand{\satprecondas}{\mbox{\upshape \textsf{sat-pre}}}
\newcommand{\remcondlessact}{\mbox{\upshape \textsf{rem-condless}}}
\providecommand{\state}{}
\renewcommand{\state}{x}
\newcommand{\statea}{x}
\newcommand{\stateb}{y}
\newcommand{\statec}{z}
\newcommand{\fals}{\mbox{\upshape \textsf{F}}}
\newcommand{\indices}{\ensuremath{V}}
\newcommand{\edges}{\ensuremath{E}}
\newcommand{\vertices}{\ensuremath{V}}
\newcommand{\listtype}{\mbox{\upshape \textsf{list}}}
\newcommand{\settype}{\mbox{\upshape \textsf{set}}}
\newcommand{\acttype}{\mbox{\upshape \textsf{action}}}
\newcommand{\graphtype}{\mbox{\upshape \textsf{graph}}}
\newcommand{\projfun}[2]{\ensuremath{\Delta_{#1}^{#2}}}
\newcommand{\snapfun}[2]{\ensuremath{\Sigma_{#1}^{#2}}}
\newcommand{\RDfun}[1]{\ensuremath{{\mathcal R}_{#1}}}
\newcommand{\elldbound}[1]{\ensuremath{{\mathcal LS}_{#1}}}
\newcommand{\distinct}{\textsf{\upshape distinct}}
\newcommand{\ddistinct}{\mbox{\upshape \textsf{ddistinct}}}
\newcommand{\simple}{\mbox{\upshape \textsf{simple}}}

\newcommand{\reachable}[3]{\ensuremath{{#1}\rightsquigarrow{#3}}}

\newcommand{\Omit}[1]{}

\newcommand{\charles}[1]{\textcolor{red}{#1}}

\newcommand{\negreachable}[3]{\ensuremath{{#2}\not\rightsquigarrow{#3}}}
\newcommand{\wdiam}[2]{{#1}^{#2}}
\newcommand{\dsnapshot}[2]{\Delta_{#1}}
\newcommand{\ellsnapshot}[2]{{\mathcal L}_{#1}}

\newcommand{\snapshotsymbol}{|\kern-.7ex\raise.08ex\hbox{\scalebox{0.7}{$\bullet$}}}
\newcommand{\snapshot}[2]{\ensuremath{\mathrel{#1\snapshotsymbol_{#2}}}}
\newcommand{\vstype}{\texttt{\upshape VS}}
\newcommand{\vtype}{{\scriptsize \ensuremath{\dom(\delta)}}}
\newcommand{\Balgo}{{\mbox{\textsc{Hyb}}}}
\newcommand{\ssgraph}[1]{\graph_\ss}
\newcommand{\agree}{\textsf{\upshape agree}}
\newcommand{\ck}{\ensuremath{\texttt{ck}}}
\newcommand{\lk}{\ensuremath{\texttt{lk}}}
\newcommand{\gr}{\ensuremath{\texttt{gr}}}
\newcommand{\gk}{\ensuremath{\texttt{gk}}}
\newcommand{\CK}{\ensuremath{\texttt{CK}}}
\newcommand{\LK}{\ensuremath{\texttt{LK}}}
\newcommand{\GR}{\ensuremath{\texttt{GR}}}
\newcommand{\GK}{\ensuremath{\texttt{GK}}}
\newcommand{\safe}{\ensuremath{\texttt{s}}}

\newcommand{\derivname}{\ensuremath{\partial}}
\newcommand{\deriv}[3]{\ensuremath{\derivname(#1,#2,#3)}}
\newcommand{\derivabbrev}[3]{\ensuremath{{\partial(#1,#2)}}}
\newcommand{\subsetoracle}{\ensuremath{ \Omega}}
\newcommand{\Aalgo}{{\mbox{\textsc{Pur}}}}
\newcommand{\Sname}{\textsf{\upshape S}}
\newcommand{\Sbrace}[1]{\Sname\langle#1\rangle}
\newcommand{\SalgoName}{\Sname_{\textsf{\upshape max}}}
\newcommand{\Salgo}[1]{\SalgoName\langle#1\rangle}

\newcommand{\WLPname}{{\mbox{\textsc{wlp}}}}
\newcommand{\WLPbrace}[1]{\WLPname\langle#1\rangle}
\newcommand{\WLPalgoName}{\WLPname_{\textsf{\upshape max}}}
\newcommand{\WLP}[1]{\WLPalgoName\langle#1\rangle}

\newcommand{\Nname}{\ensuremath{\textsf{\upshape N}}}
\newcommand{\Nbrace}[1]{\Nname\langle#1\rangle}
\newcommand{\NalgoName}{\Nname{_{\textsf{\upshape sum}}}}
\newcommand{\Nalgobrace}[1]{\NalgoName\langle#1\rangle}

\newcommand{\acycNname}{\widehat{\textsf{\upshape N}}}
\newcommand{\acycNbrace}[1]{\acycNname\langle#1\rangle}
\newcommand{\acycNalgoName}{\acycNname{_{\textsf{\upshape sum}}}}
\newcommand{\acycNalgobrace}[1]{\acycNalgoName\langle#1\rangle}

\newcommand{\Mname}{\ensuremath{\textsf{\upshape M}}}
\newcommand{\Mbrace}[1]{\Mname\langle#1\rangle}
\newcommand{\MalgoName}{\Mname{_{\textsf{\upshape sum}}}}
\newcommand{\Malgobrace}[1]{\MalgoName\langle#1\rangle}
\newcommand{\cardinality}[1]{{\ensuremath{|#1|}}}
\newcommand{\length}[1]{\cardinality{#1}}
\newcommand{\basecasefun}{\ensuremath{b}}
\newcommand{\Basecasefun}{\ensuremath{\mathcal B}}

\newcommand{\edgegen}{\ensuremath{e}}
\newcommand{\vertexgen}{\ensuremath{u}}
\newcommand{\vertexa}{{\ensuremath{\vertexgen_1}}}
\newcommand{\vertexb}{{\ensuremath{\vertexgen_2}}}
\newcommand{\vertexc}{{\ensuremath{\vertexgen_3}}}
\newcommand{\vertexd}{{\ensuremath{\vertexgen_4}}}
\newcommand{\vertexe}{{\ensuremath{\vertexgen_5}}}
\newcommand{\vertexf}{{\ensuremath{\vertexgen_6}}}
\newcommand{\vertexg}{{\ensuremath{\vertexgen_7}}}
\newcommand{\vertexsetgen}{\ensuremath{\mathit{us}}}
\newcommand{\vertexseta}{\vertexsetgen_1}
\newcommand{\vertexsetb}{\vertexsetgen_2}
\newcommand{\labelsymbol}{\ensuremath{l}}
\newcommand{\labelfun}{\ensuremath{\mathcal{L}}}
\newcommand{\DAG}{\ensuremath{A}}
\newcommand{\NalgoNameN}{{\ensuremath{\NalgoName_{\mathbb{N}}}}}
\newcommand{\NnameN}{\ensuremath{\Nname_\mathbb{N}}}
\newcommand{\replaceprojsinglename}{\raisebox{-0.3mm} {\scalebox{0.7}{\textpmhg{H}}}}
\newcommand{\replaceprojsingle}[3] {{ #2} \underset {#1} {\raisebox{-0.3mm} {\scalebox{0.7}{\textpmhg{H}}}}  #3}
\newcommand{\HOLreplaceprojsingle}[1]{\underset {#1} {\raisebox{-0.3mm} {\scalebox{0.7}{\textpmhg{H}}}}}

\newcommand{\lotus}{{\scalebox{0.6}{\includegraphics{lotus.pdf}}}}
\newcommand{\invlotus}{\mathbin{\rotatebox[origin=c]{180}{$\lotus$}}}
\newcommand{\clique}{\ensuremath{K}}
\newcommand{\partition}{\ensuremath{\vs_{1..n}}}
\newcommand{\partitiontype}{\ensuremath{\vstype_{1..n}}}
\newcommand{\vtxpartition}{\ensuremath{P}}

\newcommand{\traversalDiamAlgo}{{\mbox{\textsc{TravDiam}}}}
\newcommand{\prefix}{\textsf{\upshape pfx}}
\newcommand{\powerset}{\mathbb{P}}
\newcommand{\postfix}{\textsf{\upshape sfx}}
\newcommand{\dfunproj}{\ensuremath{{\mathfrak D}}}
\newcommand{\dfunsnap}{\ensuremath{{\textgoth D}}}
\newcommand{\ellfunproj}{\ensuremath{\mathfrak L}}
\newcommand{\ellfunsnap}{\ensuremath{\textgoth L}}
\newcommand{\cycle}{\ensuremath{C}}
\newcommand{\petal}{\ensuremath{\eta}}
\renewcommand{\prod}{\ensuremath{{{{{\mathlarger{\mathlarger {{\mathlarger {\Pi}}}}}}}}}}
\newcommand{\sccset}{{\ensuremath{SCC}}}
\newcommand{\scc}{{\ensuremath{scc}}}
\newcommand{\negate}[1]{\overline{#1}}
\newcommand{\setofsets}{\ensuremath{S}}
\newcommand{\group}{\ensuremath{\cal \Gamma}}
\newcommand{\neededvars}{{\cal N}}
\newcommand{\sspace}{\mbox{\upshape \textsf{sspc}}}
\newcommand{\tip}{\ensuremath{t}}
\newcommand{\vara}{\ensuremath{\v_1}}
\newcommand{\varb}{\ensuremath{\v_2}}
\newcommand{\varc}{\ensuremath{\v_3}}
\newcommand{\vard}{\ensuremath{\v_4}}
\newcommand{\vare}{\ensuremath{\v_5}}
\newcommand{\varf}{\ensuremath{\v_6}}
\newcommand{\varg}{\ensuremath{\v_7}}
\newcommand{\varh}{\ensuremath{\v_8}}
\newcommand{\vari}{\ensuremath{\v_9}}
\newcommand{\acta}{\ensuremath{\act_1}}
\newcommand{\actb}{\ensuremath{\act_2}}
\newcommand{\actc}{\ensuremath{\act_3}}
\newcommand{\actd}{\ensuremath{\act_4}}
\newcommand{\acte}{\ensuremath{\act_5}}
\newcommand{\actf}{\ensuremath{\act_6}}
\newcommand{\actg}{\ensuremath{\act_7}}
\newcommand{\acth}{\ensuremath{\act_8}}
\newcommand{\acti}{\ensuremath{\act_9}}

\tikzset{dots/.style args={#1per #2}{line cap=round,dash pattern=on 0 off #2/#1}}
\providecommand{\moham}[1]{\fbox{{\bf \@Mohammad: }#1}}
\newcommand{\TDbound}{{\mbox{\textsc{Arb}}}}
\newcommand{\expbound}{{\mbox{\textsc{Exp}}}}
\newcommand{\sasdom}{\expbound}
\newcommand{\cardfun}{\ensuremath{\mathbb{C}}}
\newcommand{\AGNa}{AGN1}
\newcommand{\AGNb}{AGN2}
\newcommand{\reset}{{\ensuremath{reset}}}

\newcommand{\matching}{{\cal M}}
\newcommand{\BlossomAlg}{{\mbox{\textsc{Find\_Max\_Matching}}}}
\newcommand{\AugPathAlg}{{\mbox{\textsc{Aug\_Path\_Search}}}}
\newcommand{\BlossomOrAugPath}{{\mbox{\textsc{Compute\_Blossom}}}}

%% file: isabelle.tex
\newcommand{\const}[2]{\newcommand{#1}{\textsf{#2}}}
\newcommand{\type}[2]{\newcommand{#1}{\textsf{#2}}}
\newcommand{\key}[2]{\newcommand{#1}{\textbf{#2}}}

\newcommand{\typef}[1]{\mathsf{#1}}
\newcommand{\constf}[1]{\textsf{#1}}

\newcommand{\Kzero}{\textsf{P}_\textsf{X}}
\newcommand{\Kstep}{\constf{P}_\textsf{step}}
\newcommand{\Lact}{\constf{L}_\textsf{act}}
\newcommand{\Sstep}{\constf{S}_\textsf{step}}
\newcommand{\measurable}{\mathbin{\rightarrow_M}}
\newcommand{\bindop}{\mathbin{>\!\!\!>\mkern-6.7mu=}}
\newcommand{\tendsto}{\xrightarrow{\hphantom{AAA}}}

\const{\bind}{bind}

\const{\MDPreward}{MDP-reward}
\const{\MDP}{MDP}
\const{\C}{C}

\const{\Ane}{A-ne}
\const{\subprobalgebra}{subprob-algebra}
\const{\policystep}{policy-step}
\const{\policyimprovement}{policy-improvement}
\const{\policyiteration}{policy-iteration}
\const{\streamspace}{stream-space}
\const{\completespace}{complete-space}
\const{\argmaxA}{arg-max}
\const{\hasargmax}{has-arg-max}
\const{\maxLex}{max-L-ex}
\const{\findpolicy}{find-policy}
\const{\vipolicy}{vi-policy}
\const{\conserving}{conserving}
\const{\clog}{log}
\const{\vi}{value-iteration}
\const{\improving}{improving}
\const{\pmf}{pmf}
\const{\isdec}{is-dec}
\const{\ispolicy}{is-policy}
\const{\isdecdet}{is-dec-det}
\const{\mkdecdet}{mk-dec-det}
\const{\mkstationary}{mk-stationary}
\const{\falseA}{False}
\const{\trueA}{True}
\const{\Suc}{Suc}
\const{\fst}{fst}
\const{\snd}{snd}
\const{\probspace}{$\mathcal{P}$}
\const{\tracespace}{$\mathcal{T}$}

\const{\return}{return}
\const{\cempty}{empty}

\const{\prob}{$\mathbb{P}$}
\const{\setpmf}{set-pmf}
\const{\mappmf}{map-pmf}
\const{\mapA}{map}
\const{\returnpmf}{return-pmf}
\const{\cspace}{space}
\const{\sets}{sets}
\const{\indicator}{indicator}
\const{\countspace}{count-space}
\const{\clim}{lim}
\const{\csup}{$\bigsqcup$}
\const{\policies}{$\Pi$}
\const{\vecB}{$V_B$}

\const{\bfun}{bfun}
\const{\bounded}{bounded}
\const{\range}{range}
\const{\undefined}{undefined}
\const{\dist}{$d_\infty$}
\const{\reverse}{reverse}
\const{\IT}{IT}
\const{\norm}{norm}
\const{\borel}{borel}
\const{\T}{T}
\const{\cP}{P}
\const{\distr}{distr}
\const{\univ}{UNIV}
\const{\real}{$\mathbb{R}$}
\const{\ereal}{ereal}
\const{\Up}{Up}
\const{\Right}{Right}
\const{\Down}{Down}
\const{\action}{action}
\const{\Left}{Left}
\const{\cAE}{AE}
\const{\cin}{in}
\const{\X}{X}
\const{\Y}{Y}
\const{\Pt}{$\mathcal{X}$}
\const{\id}{id}
\const{\cL}{L}
\const{\Trap}{Trap}
\const{\stateA}{state}
\const{\Pos}{Pos}
\newcommand{\LL}{\mathcal{L}}
\newcommand{\LLb}{\mathcal{L}_b}
\newcommand{\Kst}{\textsf{K}_\mathsf{st}}
\newcommand{\EK}{\mathcal{K}_\mathsf{st}}
\newcommand{\pushexp}{\textsf{pushexp}}
\newcommand{\rdec}{\textsf{r}_\mathsf{dec}}
\newcommand{\rb}{\textsf{r}_\mathsf{b}}
\const{\etr}{etr}
\newcommand{\etrfin}{\textsf{etr}_\mathsf{fin}}
\newcommand{\etropt}{\nu^*}
\newcommand{\rM}{\mathsf{r}_\mathsf{M}}
\newcommand{\TT}{\mathcal{T}}
\newcommand{\PP}{\mathcal{P}}
\newcommand{\XX}{\textsf{X}}
\newcommand{\YY}{\textsf{Y}}
\newcommand{\KKzero}{\textsf{P}_\textsf{X}}
\const{\condpmf}{cond-pmf}
\const{\asmarkovian}{as-markovian}
\const{\actstar}{act*}
\const{\finite}{finite}
\const{\isargmax}{is-arg-max}
\const{\GS}{\textsf{G}}

\newcommand{\pisuc}{\pi\textsf{-Suc}}
\newcommand{\YX}{\textsf{Y}^\textsf{X}}

\key{\klocale}{locale}
\key{\typedef}{typedef}
\key{\datatype}{datatype}
\key{\kfix}{fixes}
\const{\cfix}{fix}
\key{\kand}{and}
\key{\klet}{let}
\key{\kin}{in}
\key{\kif}{if}
\key{\then}{then}
\key{\kelse}{else}
\key{\kdo}{do}
\key{\kcase}{case}
\key{\kof}{of}
\key{\assume}{assumes}
\key{\kshow}{shows}

\type{\boolA}{bool}
\type{\nat}{$\mathbb{N}$}
\type{\set}{set}
\type{\tlist}{list}
\type{\stream}{stream}
\type{\measurepmf}{measure-pmf}
\type{\probalgebra}{prob-algebra}
\type{\vsigma}{vsigma}
\type{\streams}{streams}
\type{\cbind}{bind}
\type{\measure}{measure}
\type{\emeasure}{emeasure}
\type{\metricspace}{metric-space}
\type{\realnormedvector}{real-normed-vector}
\type{\realvector}{real-vector}
\type{\countable}{countable}
\type{\dec}{dec}
\type{\pol}{pol}

\renewcommand{\iff}{\longleftrightarrow}

%% file: intro.tex
\section{Introduction}

Matching is a classical problem in computer science, operations research, graph theory, and combinatorial optimisation.
In short, in this problem, given an undirected graph, one tries to compute a subset of the edges of this graph, s.t.\ no two edges are incident on the same vertex.
This subset is usually optimised w.r.t.\ a given objective, e.g.\ matching cardinality, sum of weights of edges in the matching, etc.
An important special case of matching problems is maximum cardinality matching in bipartite graphs.
It is one of the first problems to be addressed in combinatorial optimistation, where, for instance, the Hungarian method was invented in 1955 to solve it in the edge-weighted setting~\cite{HungarianMethodAssignment}.
The online version of that problem, i.e.\ the version in which one of the parties of the graphs arrive online, one vertex at a time, along with its incident edges, has received special attention.
This is because the problem can model many economic situations, most-notably Google's Adwords market~\cite{AdWords2007}.

The most basic version of online bipartite matching is the one where vertices and edges have no weights.
That problem was studied by Karp, Vazirani, and Vazirani (henceforth, KVV)~\cite{karpOptimalAlgorithmOnline1990}, where they devised the so-called RANKING algorithm.
In that paper, KVV showed that their algorithm can solve the online problem with a \emph{competitive ratio}, i.e.\ the average case ratio of the online algorithm's solution quality compared to the best offline algorithm, of $1-1/e$.
They also showed that this ratio is the best possible for any randomised online bipartite matching algorithm.
The analysis of the RANKING algorithm been continuously studied, where authors have mainly tried to simplify the algorithm's original correctness proof, i.e.\ the proof that it achieves a $1-1/e$ competitive ratio~\cite{GoelMehtaOnlineMatching2008,onlineMatchingSimple2008,devanurRandomizedPrimalDualAnalysis2013,onlineMatchingEcon,vaziraniOnlineMatching2022,rankingHighProbability}.
This is because the algorithm's analysis, which can be divided into a probabilistic and a combinatorial part, is considered to be ``extremely difficult''~\cite{vaziraniOnlineMatchingArxiv} by the algorithms community, despite the algorithm itself being very simple.

In this paper we formalise an analysis of the algorithm by Birnbaum and Mathieu~\cite{onlineMatchingSimple2008} (henceforth, BM) in Isabelle/HOL~\cite{DBLP:books/sp/NipkowPW02}.
BM claim to present the first simple proof of the algorithm's competitive ratio.
Indeed, the paper's title is ``Online bipartite matching made simple'', and it is the last attempt at a simple combinatorial proof for the algorithm, as later works focused on primal-dual analyses of the algorithm.

Our most striking finding is that there is a ``gap'' in the proof, where there was one lemma whose proof was ``a simple structural observation'' by the authors.
Formalising the proof of this lemma constitutes the majority of the effort that went into the work we describe here as well as the majority of the volume of the formal proof scripts.
There are also other interesting aspects, from a formalisation perspective, of that proof.
For instance, it combines graph theoretic, probabilistic, and graphical arguments.
It also requires modelling and reasoning about online algorithms.

The rest of the paper is structured as follows.
We first describe the algorithm and how we model it in Isabelle/HOL.
Then we discuss the the probabilistic part of the proof and its formalisation.
We then discuss the combinatorial part of the proof, where we describe the main findings of this work, namely, \begin{enumerate*}\item the first complete proof that covers the gap in the proof by BM, as well as other combinatorial proofs of the algorithm, and \item a significantly simpler proof of a lemma needed by BM to facilitate the algorithm's probabilistic analysis.\end{enumerate*}
Lastly, we discuss a part of the proof usually glossed over by other authors, which is lifting the analysis to obtain an asymptotic statement on the competitive ratio.

\noindent \textbf{Isabelle/HOL}
Isabelle/HOL~\cite{paulson1994isabelle} is a theorem prover based on Higher-Order Logic.
Roughly speaking, Higher-Order Logic can be seen as a combination of functional programming with logic.
Isabelle's syntax is a variation of Standard ML combined with (almost) standard mathematical notation.
Function application is written infix, and functions are usually curried (i.e., function $f$ applied to arguments $x_1~\ldots~x_n$ is written as $f~x_1~\ldots~x_n$ instead of the standard notation $f(x_1,~\ldots~,x_n)$).
In Isabelle/HOL, \textit{SOME} is the Hilbert choice, and \textit{THE} is the definite description operator.

\noindent \textbf{Availability}
Our formalisation is in the supplementary material and will be available online in case of acceptance.
Throughout the paper, and in the appendix, we added excerpts from the formalisation representing important definitions and theorem statements to aid in linking the informal description in the paper and the formal proof scripts.

%% file: definitions.tex
\renewcommand{\vertices}{\ensuremath{\mathcal{V}}}
\newcommand{\lparty}{\ensuremath{V}}
\newcommand{\rparty}{\ensuremath{U}}
\newcommand{\lperm}{\ensuremath{\sigma}}
\newcommand{\rperm}{\ensuremath{\pi}}
\renewcommand{\vertexgen}{\ensuremath{v}}
\newcommand{\lvertexgen}{\ensuremath{v}}
\newcommand{\lvertexa}{\lvertexgen_1}
\newcommand{\lvertexb}{\lvertexgen_2}
\newcommand{\lvertexc}{\lvertexgen_3}
\newcommand{\lvertexd}{\lvertexgen_4}
\newcommand{\lvertexe}{\lvertexgen_5}
\newcommand{\lvertexf}{\lvertexgen_6}
\newcommand{\rvertexgen}{\ensuremath{u}}
\newcommand{\rvertexa}{\rvertexgen_1}
\newcommand{\rvertexb}{\rvertexgen_2}
\newcommand{\rvertexc}{\rvertexgen_3}
\newcommand{\rvertexd}{\rvertexgen_4}
\newcommand{\rvertexe}{\rvertexgen_5}
\newcommand{\rvertexf}{\rvertexgen_6}
\newcommand{\lorder}{\ensuremath{\pi}}
\newcommand{\rorder}{\ensuremath{\sigma}}
\newcommand{\rank}{\textit{online-match}}
\newcommand{\neighb}[2]{\ensuremath{{N_{#1} ({#2})}}}
\newcommand{\shiftsto}{\textit{shifts-to}}
\newcommand{\zig}{\textit{zig}}
\newcommand{\zag}{\textit{zag}}
\newcommand{\lpartyperm}{\ensuremath{\lparty'}}

\newcommand{\isaname}[1]{\emph{#1}}

\newcommand{\nth}[2]{#1[#2]}
\newcommand*{\uniform}{\mathcal{U}}
\newcommand*{\rankingprob}{\textit{RANKING}}
\newcommand*{\perms}{\mathcal{S}}
\newcommand*{\bernoulli}{\mathbb{I}}
\DeclarePairedDelimiter{\card}{|}{|}

\section{Basic Definitions and Notation}
We denote a list of elements as $[x_1,x_2,\dots,x_n]$.
In the rest of this paper, we only consider lists with distinct elements.
We say element $x_i$ has rank $i$\footnote{In the formalisation we use index, which is the same as the rank less one.} in the list $[x_1,x_2,\dots,x_i,\dots,x_n]$.
We overload the membership, subset, union and intersection set operations to lists.
For a list $\vs$, of length $n$, and an element $\vertexgen \in \vs$, let, for $1\leq i\leq n$, $\vs[\vertexgen \mapsto i]$ denote the list with the same elements as $\vs$, where $\vertexgen$ has rank $i$, the elements of rank less than $i$ remain unchanged and the
rank of the elements of rank at least $i$ is increased by $1$.
Also, let $\vs(\vertexgen)$ denote the rank of $\vertexgen$ in $\vs$ and $\nth{\vs}{i}$ the element of rank $i$ in $\vs$.
For a list $\vs$, 
$\vertexgen\concat\vs$ denotes the list $\vs$ but with the vertex $\vertexgen$ appended to its head.
A permutation of a finite set $s$ is a list whose elements are exactly the elements of $s$.

An edge is a set of vertices with size 2.
A graph $\graph$ is a set of edges.
The set of vertices of a graph $\graph$, denoted by $\vertices(\graph)$, is $\bigcup_{\edgegen\in\graph}\edgegen$.
For a vertex $\lvertexgen$, $\neighb{\graph}{\lvertexgen}$ denotes $\{\rvertexgen\mid\{\lvertexgen,\rvertexgen\}\in\graph\}$.
We say a graph $\graph$ is bipartite w.r.t.\ to two sets of vertices $\lparty$ and $\rparty$ (henceforth, the left and right party) iff \begin{enumerate*}\item $\vertices(\graph) \subseteq (\lparty \cup \rparty)$, \item for any $\{\lvertexgen,\rvertexgen\}\in\graph$, we have that $\{\lvertexgen,\rvertexgen\}\not\subseteq\lparty$ and $\{\lvertexgen,\rvertexgen\}\not\subseteq\rparty$.\end{enumerate*}
A set of edges $\matching$ is a matching iff
$\forall e\neq e'\in\matching.\; e \cap e' = \emptyset$. 
For a matching $\matching$ and a vertex $\lvertexgen$, if there is $\rvertexgen$ s.t.\ $\{\lvertexgen,\rvertexgen\}\in\matching$, we say $\rvertexgen$ is the partner of $\lvertexgen$, denoted by $\matching(\lvertexgen)$.
We use $\graph-E$ to denote the edges in $\graph$ that are not in $E$, and, for a set of vertices $V$, $\graph\setminus V$ denotes $\graph\cap\{\edgegen\mid \edgegen\cap V = \emptyset\}$, i.e.\ the graph with edges incident to vertices in $V$ removed.


In many cases, a matching is a subset of a graph, in which case we call it a matching w.r.t. the graph.
For a graph $\graph$, a matching $\matching$ w.r.t\ $\graph$ is a maximum cardinality matching, aka maximum matching, w.r.t.\ $\graph$ iff for any matching $\matching'$ w.r.t. $\graph$, we have that $\cardinality{\matching'} \leq  \cardinality{\matching}$.
A matching $\matching$ w.r.t.\ $\graph$ is a perfect matching w.r.t.\ $\graph$ iff $\vertices(\matching)=\vertices(\graph)$.
A matching $\matching$ w.r.t.\ $\graph$ is a maximal matching w.r.t.\ $\graph$ iff
$\forall\edgegen\in\graph.\ \edgegen\cap\vertices(\matching)\neq\emptyset$.

A discrete probability space $P$ is defined by a countable sample space $\Omega_P$ and a
probability mass function (PMF) $\prob_P : \Omega_P \to [0,1]$ assigning a
probability to each sample, where $\sum_{\omega \in \Omega_P} \prob_P(\omega) = 1$.
The PMF is lifted naturally to events (sets of samples) as
$\prob_P(E) = \sum_{\omega \in E} \prob_P(E)$ for $E \subseteq \Omega_P$. The expectation
of a random variable $X : \Omega_P \to \real$ is denoted
$\mathbb{E}_{\omega \sim P} \left[ X(\omega) \right]$.
For a set $B$ and a non-empty, finite subset $A \subseteq B$, $\uniform_B(A)$ is
the discrete uniform distribution, \ie $\Omega_{\uniform_B(A)} = B$ and
$\prob_{\uniform_B(A)}(a) = \frac{1}{\card{A}}$ if $a \in A$ and
$\prob_{\uniform(A)}(b) = 0$ if $b \notin A$. If $A = B$ we simply write
$\uniform(A)$ for $\uniform_A(A)$.

We model randomised algorithms as probability distributions over the results of the algorithm.
The Giry Monad~\cite{Giry80} allows to compose random experiments in an elegant manner and is used to express randomised algorithms. 
The $\return$ operator gives a distribution which places probability $1$ on a single sample $\omega$, \ie $\prob_{\return(\omega)}(x)$ is $1$, if $x = \omega$, and $0$, otherwise.

Composition of experiments is achieved via the $\cbind$ operator (written infix as $\bindop$).
Intuitively, $P \bindop Q$ randomly chooses a sample $\omega$ according to $P$ and then returns a value chosen randomly according to the distribution $Q(\omega)$.
For additional legibility, we use Haskell-like $\kdo$-notation for $\cbind$ and $\return$.
This notation can be desugared recursively as follows:
\begin{align*}
  & \kdo\{~x \leftarrow P ;~\textit{stmts}~\} \equiv P \bindop (\lambda x.~\textit{stmts}).
\end{align*}

In Isabelle/HOL, we base our work on a simple formalisation of undirected graphs by Abdulaziz et al.~\cite{DBLP:conf/mfcs/AbdulazizMN19}, which was introduced in the context of the verification of Edmonds' blossom matching algorithm.
The types of graphs and edges as well as the notion a matching in this formalisation are shown in Listing~\ref{isa:graph_matching}.
We use this formalisation because of its simplicity, and the fact that it has a rich library on matchings and other related notions, as we will discuss later.
However, we will not further discuss the merits of this representation as it is outside of the scope of this work.
Interested readers should consult the original paper~\cite{DBLP:conf/mfcs/AbdulazizMN19}.

Probability theory in Isabelle/HOL is based on a general formalisation of measure theory by Hölzl~\cite{Holzl13}.
In the formalisation, $\uniform(A)$ is denoted \isaname{pmf\_of\_set\ $A$}, and $\return$ is denoted \isaname{return\_pmf}.
The meanings of other Isabelle/HOL notations used in the rest of the paper should be self-explanatory.

\section{RANKING}
\input{rankingfig}
\newcommand{\matchededgecolor}{green}
\input{RANKINGalgo}
Given a bipartite graph $\graph$, whose left and right parties are $\lparty$ and $\rparty$, the ranking algorithm takes as an offline input $\lparty$, and a sequence $\rperm$ as an online input, where vertices, along with their adjacent edges, arrive one-by-one.
As an example, consider Fig.~\ref{fig:offline_start}, showing a graph whose left
party, i.e.\ the offline vertices, is
$\{\lvertexd, \lvertexb, \lvertexf, \lvertexa, \lvertexe, \lvertexc \}$. The
right party, i.e.\ the online vertices, arrive in the order
$[\rvertexa, \rvertexb, \rvertexc, \rvertexd, \rvertexe]$.
The first step in the algorithm is that it randomly permutes the offline input.
In our example, this is shown in Fig.~\ref{fig:offline_perm}.
Then, vertices from the right party of the graph arrive one-by-one.
The most important thing to note about that is that, for every arriving vertex $\rvertexgen$, the algorithm adds the edge connecting $\rvertexgen$ and the offline unmatched vertex with the minimum rank, if any such edge exists.
In our example, we have the ranking
$[\lvertexa, \lvertexb, \lvertexc, \lvertexd, \lvertexe, \lvertexf]$, of the offline vertices.
Fig.~\ref{fig:online_a} shows the state of the matching after the arrival of
$\rvertexa$: it has three edges connecting it to the offline vertices $\lvertexa$, $\lvertexc$, and $\lvertexe$.
The edge connecting it to $\lvertexa$ is added to the matching, as it
is unmatched and has the lowest rank among them.
Then, the other vertices on the online side arrive based on the order given earlier, and the matching is updated, as shown in Fig.~\ref{fig:online_b}-\ref{fig:online_e}, and the final matching computed by the algorithm is the one represented by the \matchededgecolor{} edges in Fig.~\ref{fig:online_f}.

As should be clear by now, the algorithm's description and, accordingly, modeling is a simple task.
The pseudo-code is in Algorithm~\ref{alg:ranking}.
In Isabelle/HOL, we model the algorithm as shown in Listing~\ref{isa:rank}.
The first two functions are recursive on lists.
The first function, \isaname{step}, is recursive on the list of the offline vertices, where, given a graph \isaname{G}, a vertex \isaname{\rvertexgen} from the online side, the list of offline vertices, and the matching, it adds to the matching the first edge it finds that connects \isaname{\rvertexgen} and an offline vertex \isaname{\lvertexgen}.
The function does the recursion on the list, assuming the list is ordered
according to the ranking of the offline vertices, with the head of the list
being the vertex with the lowest rank.
The second function, \isaname{online\_match'}, is recursive on the on the list of online vertices, where the list is ordered according to the arrival order of those vertices, where the head of the list is the earliest arriving vertex.
For each vertex in the list, \isaname{online\_match'} tries to match it to an offline vertex using \isaname{step}.
The other main function, \isaname{ranking}, chooses a permutation of the offline vertices and passes it to \isaname{online-match}.

\input{ranking_listing}

We note that we avoid devising an involved way to model and reason about online computation, and only model it simply as a list of inputs and a step function that operates on each online input.
This is because the algorithm description itself is simple.
The primary focus of our work here is the formalisation of the correctness argument, the mathematical part of which is the main challenge.

\subsection{Competitive Ratio of \rankingprob}

The goal of this work is formalise the analysis of \rankingprob's competitiveness.
In general, for online algorithms solving optimisation problems, the analysis focuses on the quality of their outputs in comparison with the quality of the output of the best offline algorithm, i.e.\ an algorithm which has access to the entire input before it starts computing its output.
The outcome of such an analysis is referred to as the \emph{competitive ratio} of the respective online algorithm.
In the case of bipartite matchings, the best offline algorithms, like the
Hopcroft-Karp algorithm~\cite{HopcroftKarp1973}, can compute maximum cardinality matching for bipartite graphs.
Thus, for \rankingprob, the natural way to analyse it is by showing that the size of the matching it computes maintains a certain ratio if compared to the size of the maximum matching of the input graph.
Furthermore, since \rankingprob\ is a randomised algorithm, it is natural that this relationship is in expectation.
More precisely, for \rankingprob, we have the following relation, which was
first shown by KVV: for any given graph and arrival orders, the ratio between the expected size of the matching computed by \rankingprob\ and the size of the maximum matching is $1-1/e$.
The expectation ranges over the different permutations of the offline side.

%% file: rankingfig.tex
\begin{figure}[t]
  \begin{subfigure}[t]{0.15\textwidth}
  \begin{tikzpicture}[rotate=0,yscale=1,xscale=-1]
    \input{offline_start}
  \end{tikzpicture}
  \caption{\label{fig:offline_start}}
  \end{subfigure}
  \begin{subfigure}[t]{0.15\textwidth}
  \begin{tikzpicture}[rotate=0,yscale=1,xscale=-1]
    \input{offline_perm}
  \end{tikzpicture}
  \caption{\label{fig:offline_perm}}
  \end{subfigure}
  \begin{subfigure}[t]{0.15\textwidth}
  \begin{tikzpicture}[rotate=0,yscale=1,xscale=-1]
    \input{online_a}
  \end{tikzpicture}
  \caption{\label{fig:online_a}}
  \end{subfigure}
  \begin{subfigure}[t]{0.15\textwidth}
  \begin{tikzpicture}[rotate=0,yscale=1,xscale=-1]
    \input{online_b}
  \end{tikzpicture}
  \caption{\label{fig:online_b}}
  \end{subfigure}
  \begin{subfigure}[t]{0.15\textwidth}
  \begin{tikzpicture}[rotate=0,yscale=1,xscale=-1]
    \input{online_c}
  \end{tikzpicture}
  \caption{\label{fig:online_c}}
  \end{subfigure}
  \begin{subfigure}[t]{0.15\textwidth}
  \begin{tikzpicture}[rotate=0,yscale=1,xscale=-1]
    \input{online_d}
  \end{tikzpicture}
  \caption{\label{fig:online_d}}
  \end{subfigure}
  \begin{subfigure}[t]{0.15\textwidth}
  \begin{tikzpicture}[rotate=0,yscale=1,xscale=-1]
    \input{online_e}
  \end{tikzpicture}
  \caption{\label{fig:online_e}}
  \end{subfigure}
  ~~~~
  \begin{subfigure}[t]{0.15\textwidth}
  \begin{tikzpicture}[rotate=0,yscale=1,xscale=-1]
    \input{online_f}
  \end{tikzpicture}
  \caption{\label{fig:online_f}}
  \end{subfigure}
  ~~~~
  \begin{subfigure}[t]{0.2\textwidth}
  \begin{tikzpicture}[rotate=0,yscale=1,xscale=-1]
    \input{online_g}
  \end{tikzpicture}
  \caption{\label{fig:online_g}}
  \end{subfigure}
  \begin{subfigure}[t]{0.2\textwidth}
  \begin{tikzpicture}[rotate=0,yscale=1,xscale=-1]
    \input{zig_zag_fig}
  \end{tikzpicture}
  \caption{\label{fig:online_zig}}
  \end{subfigure}
  \begin{subfigure}[t]{0.2\textwidth}
  \begin{tikzpicture}[rotate=0,yscale=1,xscale=1]
    \input{zig_zag_symm_fig}
  \end{tikzpicture}
  \caption{\label{fig:online_zag}}
  \end{subfigure}
\caption{\label{fig:ranking} The steps of computing a matching using \rank, and what happens when an online vertex is removed from the input.}
\end{figure}

%% file: offline_start.tex
    {\node (ghost) at (0,0) [varnode] {\textcolor{white}{$\vertexa$}} ;}
    {\node (b1) at (2,0) [varnode] {\textcolor{black}{$\lvertexd$}} ;}
    {\node (b2) at (2,-1) [varnode] {\textcolor{black}{$\lvertexb$}} ;}
    {\node (b3) at (2,-2) [varnode] {\textcolor{black}{$\lvertexf$}} ;}
    {\node (b4) at (2,-3) [varnode] {\textcolor{black}{$\lvertexa$}} ;}
    {\node (b5) at (2,-4) [varnode] {\textcolor{black}{$\lvertexe$}} ;}
    {\node (b6) at (2,-5) [varnode] {\textcolor{black}{$\lvertexc$}} ;}

%% file: offline_perm.tex
    {\node (ghost) at (0,0) [varnode] {\textcolor{white}{$\vertexa$}} ;}
    {\node (b1) at (2,0) [varnode] {\textcolor{black}{$\lvertexa$}} ;}
    {\node (b2) at (2,-1) [varnode] {\textcolor{black}{$\lvertexb$}} ;}
    {\node (b3) at (2,-2) [varnode] {\textcolor{black}{$\lvertexc$}} ;}
    {\node (b4) at (2,-3) [varnode] {\textcolor{black}{$\lvertexd$}} ;}
    {\node (b5) at (2,-4) [varnode] {\textcolor{black}{$\lvertexe$}} ;}
    {\node (b6) at (2,-5) [varnode] {\textcolor{black}{$\lvertexf$}} ;}

%% file: online_a.tex
    \input{offline_perm}
    \node (a1) at (0,0) [varnode] {\textcolor{black}{$\rvertexa$}} ;
    \draw [green,-,matchededge] (a1) -- (b1) ;
    \draw [-,matchededge] (a1) -- (b3) ;
    \draw [-,matchededge] (a1) -- (b5) ;

%% file: online_b.tex
    \input{online_a}
    \node (a2) at (0,-1) [varnode] {\textcolor{black}{$\rvertexb$}} ;

    \draw [-,matchededge] (a2) -- (b1) ;
    \draw [green,-,matchededge] (a2) -- (b2) ;
    \draw [-,matchededge] (a2) -- (b4) ;

%% file: online_c.tex
    \input{online_a}
    \input{online_b}

    \node (a3) at (0,-2) [varnode] {\textcolor{black}{$\rvertexc$}} ;

    \draw [-,matchededge] (a3) -- (b2) ;
    \draw [-,matchededge] (a3) -- (b1) ;
    \draw [green,-,matchededge] (a3) -- (b4) ;

%% file: online_d.tex
   \input{online_c}

   \node (a4) at (0,-3) [varnode] {\textcolor{black}{$\rvertexd$}} ;

    \draw [-,matchededge] (a4) -- (b1) ;
    \draw [green,-,matchededge] (a4) -- (b3) ;

%% file: online_e.tex
    \input{online_d}

    {\node (a5) at (0,-4) [varnode] {\textcolor{black}{$\rvertexe$}} ;}

    \draw [-,matchededge] (a5) -- (b4) ;
    \draw [green,-,matchededge] (a5) -- (b5) ;

%% file: online_f.tex
    \input{online_e}

    {\node (a6) at (0,-5) [varnode] {\textcolor{black}{$\rvertexf$}} ;}

    {
    \draw [-,matchededge] (a6) -- (b5) ;}

%% file: online_g.tex
    {\node (c1) at (0,0) [varnode] {\textcolor{black}{$\rvertexa$}} ;}
    {\node (c2) at (0,-1) [varnode] {\textcolor{red}{$\rvertexb$}} ;}
    {\node (c3) at (0,-2) [varnode] {\textcolor{black}{$\rvertexc$}} ;}
    {\node (c4) at (0,-3) [varnode] {\textcolor{black}{$\rvertexd$}} ;}
    {\node (c5) at (0,-4) [varnode] {\textcolor{black}{$\rvertexe$}} ;}
    {\node (c6) at (0,-5) [varnode] {\textcolor{black}{$\rvertexf$}} ;}

    \input{offline_perm_nov1}

    \draw [green,-,matchededge] (c1) -- (b1) ;
    \draw [-,matchededge] (c1) -- (b3) ;
    \draw [-,matchededge] (c1) -- (b5) ;

    \draw [-,matchededge] (c2) -- (b1) ;
    \draw [blue,-,matchededge] (c2) -- (b2) ;
    \draw [-,matchededge] (c2) -- (b4) ;

    \draw [red,-,matchededge] (c3) -- (b2) ;
    \draw [-,matchededge] (c3) -- (b1) ;
    \draw [blue,-,matchededge] (c3) -- (b4) ;

    \draw [-,matchededge] (c4) -- (b1) ;
    \draw [-,matchededge] (c4) -- (b3) ;

    \draw [red,-,matchededge] (c5) -- (b4) ;
    \draw [blue,-,matchededge] (c5) -- (b5) ;

    \draw [red,-,matchededge] (c6) -- (b5) ;

%% file: zig_zag_fig.tex
    {\node (c1) at (0,0) [varnode] {\textcolor{black}{$\rvertexa$}} ;}
    {\node (c2) at (0,-1) [varnode] {\textcolor{red}{$\rvertexb$}} ;}
    {\node (c3) at (0,-2) [varnode] {\textcolor{black}{$\rvertexc$}} ;}
    {\node (c4) at (0,-3) [varnode] {\textcolor{black}{$\rvertexd$}} ;}
    {\node (c5) at (0,-4) [varnode] {\textcolor{black}{$\rvertexe$}} ;}
    {\node (c6) at (0,-5) [varnode] {\textcolor{black}{$\rvertexf$}} ;}

    \input{offline_perm_nov1}

    \draw [gray,loosely dotted,-,matchededge] (c1) -- (b3) ;
    \draw [gray,loosely dotted,-,matchededge] (c1) -- (b5) ;

    \draw [gray,loosely dotted,-,matchededge] (c2) -- (b1) ;
    \draw [blue,->,matchededge] (c2) -- (b2) node[midway,above] {\scriptsize$\zig(\rvertexb)$};
    \draw [gray,loosely dotted,-,matchededge] (c2) -- (b4) ;

    \draw [red,<-,matchededge] (c3) -- (b2) node[midway,above] {\scriptsize$\zag(\lvertexb)$};
    \draw [gray,loosely dotted,-,matchededge] (c3) -- (b1) ;
    \draw [blue,->,matchededge] (c3) -- (b4) node[midway,above] {\scriptsize$\zig(\rvertexc)$};

    \draw [gray,loosely dotted,-,matchededge] (c4) -- (b1) ;
    \draw [gray,loosely dotted,-,matchededge] (c4) -- (b3) ;

    \draw [red,<-,matchededge] (c5) -- (b4) node[midway,above] {\scriptsize$\zag(\lvertexd)$};
    \draw [blue,->,matchededge] (c5) -- (b5) node[midway,above] {\scriptsize$\zig(\rvertexe)$};

    \draw [red,<-,matchededge] (c6) -- (b5) node[midway,above] {\scriptsize$\zag(\lvertexe)$};

%% file: zig_zag_symm_fig.tex
    {\node (c1) at (0,0) [varnode] {\textcolor{black}{$\rvertexa$}} ;}
    {\node (c2) at (0,-1) [varnode] {\textcolor{gray}{$\rvertexb$}} ;}
    {\node (c3) at (0,-2) [varnode] {\textcolor{black}{$\rvertexc$}} ;}
    {\node (c4) at (0,-3) [varnode] {\textcolor{black}{$\rvertexd$}} ;}
    {\node (c5) at (0,-4) [varnode] {\textcolor{black}{$\rvertexe$}} ;}
    {\node (c6) at (0,-5) [varnode] {\textcolor{black}{$\rvertexf$}} ;}

    \input{offline_perm_nov1}

    \draw [gray,loosely dotted,-,matchededge] (c1) -- (b3) ;
    \draw [gray,loosely dotted,-,matchededge] (c1) -- (b5) ;

    \draw [gray,loosely dotted,-,matchededge] (c2) -- (b1) ;
    \draw [gray,loosely dotted,-,matchededge] (c2) -- (b2) ;
    \draw [gray,loosely dotted,-,matchededge] (c2) -- (b4) ;

    \draw [gray,loosely dotted,-,matchededge] (c3) -- (b1) ;
    \draw [red,<-,matchededge] (c3) -- (b2)  node[midway,above] {\scriptsize$\zig(\lvertexb)$};
    \draw [blue,->,matchededge] (c3) -- (b4) node[midway,above] {\scriptsize$\zag(\rvertexc)$};

    \draw [gray,loosely dotted,-,matchededge] (c4) -- (b1) ;
    \draw [gray,loosely dotted,-,matchededge] (c4) -- (b3) ;

    \draw [red,<-,matchededge] (c5) -- (b4) node[midway,above] {\scriptsize$\zig(\lvertexd)$};
    \draw [blue,->,matchededge] (c5) -- (b5) node[midway,above] {\scriptsize$\zag(\rvertexe)$};

    \draw [red,<-,matchededge] (c6) -- (b5) node[midway,above] {\scriptsize$\zig(\lvertexe)$};

%% file: RANKINGalgo.tex
\SetAlgoSkip{}
\begin{algorithm}
    \SetKwData{Left}{left}\SetKwData{This}{this}\SetKwData{Up}{up}
    \SetKwFor{RepeatInf}{for}{do}{}
    \SetKwFor{Procedure}{function}{begin}{end}

    \DontPrintSemicolon
    \Procedure{\rank(\graph, \rperm, \lperm)}{
                             $\matching\gets\emptyset$\\
    \RepeatInf(){every arriving vertex $\rvertexgen$ in $\rperm$}{

\lIf{$\exists\lvertexgen\in(\neighb{\graph}{\rvertexgen}-\vertices(\matching))$}{$\matching\gets \matching \cup \{\{\argmin_{\lvertexgen\in(\neighb{\graph}{\rvertexgen}-\vertices(\matching))} \lperm(\lvertexgen),\rvertexgen\}\}$}
    }
    \Return $\matching$}
    \Procedure{\rankingprob(\graph, \rperm)}{
        $\lperm \gets$ a random permutation of $\lparty$\\
        \Return $\rank(\graph, \rperm, \lperm)$
    }
    \caption{Pseudo-code of $\rankingprob$}
    \label{alg:ranking}
\end{algorithm}

%% file: ranking_listing.tex
\begin{figure*}[t]
\begin{IsabelleSnippet}[label=isa:rank]{Modelling \rankingprob\ in Isabelle/HOL.}
fun step :: "'a graph \<Rightarrow> 'a \<Rightarrow> 'a list \<Rightarrow> 'a graph \<Rightarrow> 'a graph" where
  "step _ _ [] M = M"
| "step G u (v#vs) M = (
      if v \<notin> Vs M \<and> u \<notin> Vs M \<and> {u,v} \<in> G
      then insert {u,v} M
      else step G u vs M
    )"

fun online_match' :: "'a graph \<Rightarrow> 'a list \<Rightarrow> 'a list \<Rightarrow> 'a graph \<Rightarrow> 'a graph" where
  "online_match' _ [] _ M = M"
| "online_match' G (u#us) \<sigma> M = online_match' G us \<sigma> (step G u \<sigma> M)"

abbreviation "online_match G \<pi> \<sigma> \<equiv> online_match' G \<pi> \<sigma> {}"

definition "ranking \<equiv>
  do {
    \<sigma> \<leftarrow> pmf_of_set (permutations_of_set V);
    return_pmf (online_match G \<pi> \<sigma>)
  }"
\end{IsabelleSnippet}
\end{figure*}

%% file: probabilistic.tex
\section{Competitiveness for Bi-Partite Graphs with Perfect Matchings}

In the following, let $\graph$ be a bipartite graph w.r.t.\ $\lparty$ and
$\rparty$, s.t.\ $\matching$ is a perfect matching w.r.t.\ $\graph$, and
$\cardinality{\matching} = n$.
Let $\rperm$ be an arrival order for $\rparty$ and let $\perms(A)$ denote the
set of all permutations of a finite set $A$.\footnote{In the formalisation $\perms(A)$ is
written \isaname{permutations\_of\_set\ A}.}

The algorithm can be modelled as the following Giry monad
\[
  \rankingprob(\graph, \rperm) \equiv \kdo\ \{\ \lperm \leftarrow \uniform(\perms(\lparty));\ \return\ \rank(\graph, \rperm, \lperm)\ \}.
\]

In the following, we describe our formal proof of the analysis of the competitive ratio for instances with perfect matching.
This formal proof closely follows the one by BM.
However, we highlight the differences to the original one as they arise.

We need the following lemma (\cite[Lemma~5]{onlineMatchingSimple2008}) before the main result can be shown.
\begin{mylem}
  \label{lem:ranktmatchedbound}
  Let $x_t$ denote the probability over the random permutations of $\lparty$
  that the vertex of rank $t$ is matched by the algorithm, for $1 \leq t \leq n$. Then
  $1 - x_t \leq (1/n) \sum_{1\leq s \leq t} x_s$.
\end{mylem}
Let $\lvertexgen \in \lparty$ be the vertex of rank $t$ for some fixed permutation $\lperm$ of $\lparty$.
The intuition behind this bound is that $v$ only remains unmatched if its partner $\matching(\lvertexgen)$ in the perfect matching is matched to a vertex ranked lower in $\rperm$.
Since $\lvertexgen$ is a random vertex (when drawing a permutation), so is $\matching(\lvertexgen)$.
The right-hand-side is supposed to be the probability that $\matching(\lvertexgen)$ is matched to a vertex arriving before $\lvertexgen$ (since the sum is the expected number of vertices matched to vertices of rank at most $t$).
This intuitive idea does not work due to the dependence of $\matching(\lvertexgen)$ and the set of vertices matched to vertices of rank at most $t$.
The correct argument avoids this dependence.
However, this requires a stronger statement on what happens with $\matching(\lvertexgen)$ if $\lvertexgen$ stays unmatched, captured in the following lemma (\cite[Lemma 4]{onlineMatchingSimple2008}), whose proof we discuss in the next section.
\begin{mylem}
  \label{lem:rankunmatchedmoveto}
  Let $\lvertexgen \in \lparty$, $\rvertexgen$ denote $\matching(v)$, and $\lperm \in \perms(V)$.
  If $\lvertexgen$ is not matched by $\rank(\graph, \lperm, \rperm)$ to $\rvertexgen$, then, for all $1 \leq i \leq n$, $\rvertexgen$ is matched by
  $\rank(\graph, \lperm[\vertexgen \mapsto i], \rperm)$ to a $\lvertexgen_i \in \lparty$ s.t.\ $\lperm[\lvertexgen \mapsto i](\lvertexgen_i) \leq \lperm(\lvertexgen)$.
\end{mylem}

Before presenting the proof of \autoref{lem:ranktmatchedbound}, we need to consider how to formally
define $x_t$. It cannot be stated as a probability in the distribution
$\rankingprob(\graph, \rperm)$. There is no way to refer to the
``vertex of rank $t$ in the permutation $\lperm$'' since $\rankingprob(\graph, \rperm)$ is a distribution over
subgraphs of $\graph$ and the random permutations used to obtain them are
not accessible. The solution is to explicitly define the Bernoulli distribution
capturing the notion of the vertex of rank $t$ being matched.
\[
  \bernoulli_t \equiv \kdo\ \{\ \lperm \leftarrow \uniform(\perms(\lparty));\ \klet\ R =
  \rank(\graph, \rperm, \lperm);\ \return\ (\nth{\lperm}{t} \in \vertices(R))\ \}
\]
Then, $1 - x_t$ corresponds to the probability $\prob_{\bernoulli_t}(\falseA)$.

A key step to achieve the independence of the involved events revolves around
not only drawing a random permutation, but also drawing a random vertex and
moving it to rank $t$. This is reflected in the distribution $\bernoulli'_t$, given
in Fig.~\ref{fig:rankunmatchedmovetomonada}.
This deceptively simple change ensures the independence of the drawn permutation, i.e.\ $\lperm$, and the actual partner in the perfect matching of the vertex of rank $t$, i.e.\ $\matching(\nth{\sigma[\lvertexgen\mapsto t]}{t})$ which is the same as $\matching(\lvertexgen)$.
There is an aspect that is glossed over in the original proof and is intuitively clear: simply drawing a random permutation uniformly at random and the modified way where a random vertex is put at rank $t$ are equivalent.\todo{why are they equiv?}
This is shown explicitly in the formal proof.

\input{rank_unmatched_moveto_monads_fig}

The final distribution we present here, $\bernoulli''_t$ in
Fig.~\ref{fig:rankunmatchedmovetomonadb}, captures the probability that the
partner $\matching(\lvertexgen)$ of a random $\lvertexgen \in \lparty$ is
matched to a vertex of rank at most $t$.\footnote{The formalisations of the different distributions are in Listing~\ref{isa:do_programs}.}

\begin{proof}[Proof of \autoref{lem:ranktmatchedbound}]
  The first step follows from the fact that the permutation $\sigma$, in both $\bernoulli_t$ and $\bernoulli'_t$, and the vertex $\lvertexgen$ are all drawn from uniform distributions.
  \begin{align*}
    \prob_{\bernoulli_t}(\falseA) &= \prob_{\bernoulli'_t}(\falseA) \\
    \intertext{By \autoref{lem:rankunmatchedmoveto}, if $\lvertexgen \in \lparty$ is unmatched in $\rank(\graph, \rperm, \lperm[\lvertexgen \mapsto t])$,
    then, $\matching(\lvertexgen)$ is matched to a vertex of rank at most $t$ in $\rank(\graph, \rperm, \lperm)$ (by using $\lperm[\lvertexgen \mapsto t][\lvertexgen \mapsto \lperm(\lvertexgen)] = \lperm$).}
                                &\leq \prob_{\bernoulli''_t}(\trueA) \\
    \intertext{Then, the process of drawing a random $\lvertexgen \in \lparty$ and considering $\matching(\lvertexgen)$ in $\bernoulli''_t$
    can be replaced with drawing a random $\rvertexgen \in \rparty$ directly, using the bijection induced by $\matching$.
    This describes the probability that a random $\rvertexgen \in \rparty$ is matched to a vertex of rank at most $t$.
    That probability, in turn, is exactly the expected size of the set of online vertices matched to vertices of rank at most $t$.
    Formally, these two steps are performed by defining two more Bernoulli distributions capturing the involved concepts. Their
    definitions are omitted here. Let $\bernoulli^*_{t}$ be the distribution for the set of online vertices matched to vertices of rank at most $t$.
    }
                                &= \frac{1}{n} \mathbb{E}_{O \sim \bernoulli^*_t}[\card{O}] \\
    \intertext{The final step is to express the expected size of the set of online vertices matched to vertices of rank at most $t$
    as a sum of the probabilities that the offline vertices of rank up to $t$ are matched. This completes the argument.}
    &= \frac{1}{n} \sum_{s=1}^t \prob_{\bernoulli_s}(\trueA)
  \end{align*}
\end{proof}

Then, we proceed to the main result of this section.
\begin{mythm}
  \label{thm:rankcorrectperfect}
  The competitive ratio of $\rankingprob$ for instances with a perfect matching
  of size $n$ is at
  least $1-(1 - \frac{1}{n+1})^n$, \ie $1-(1 - \frac{1}{n+1})^n \leq \frac{\mathbb{E}_{R \sim \rankingprob(\graph, \rperm)} \left[ \card{R} \right]}{n}$.
\end{mythm}
\begin{proof}
  The expected size of the matching produced by $\rankingprob(\graph, \rperm)$ can
  be rewritten as a sum of the probabilities of the vertices of some rank
  getting matched.
  \begin{align*}
    \frac{\mathbb{E}_{R \sim \rankingprob(\graph, \rperm)}\left[ \card{R} \right]}{n} &=
    \frac{1}{n} \sum_{s=1}^n \prob_{\bernoulli_s}(\trueA) \\
    \intertext{The bound obtained on $\prob_{\bernoulli_s}(\falseA)$ for $1 \leq s \leq n$ in \autoref{lem:ranktmatchedbound} can be used
    to bound the sum. This requires a fact on sums provable by induction on $n$, followed by algebraic manipulation.}
                                                                                                 &\geq \frac{1}{n} \sum_{s=1}^n \left( 1 - \frac{1}{n+1} \right)^s \\
    \intertext{More algebraic manipulation yields the final result.}
                                                                                                 &= 1 - \left( 1 - \frac{1}{n+1} \right)^n
  \end{align*}
\end{proof}


%% file: rank_unmatched_moveto_monads_fig.tex
\begin{figure}[t]
  \begin{subfigure}[t]{.45\textwidth}
    \begin{flalign*}
      \bernoulli'_t \equiv{} &\kdo\ \{\ &\\
                                    & \sigma \leftarrow \uniform(\perms(\lparty)); \\
                                    & \lvertexgen \leftarrow \uniform(\lparty);\\
                                    &\klet\ R = \rank(\graph, \pi, \sigma[\lvertexgen \mapsto t]);\\
                                    & \return\ (\lvertexgen \in \vertices(R))\\
      \}
    \end{flalign*}
    \subcaption{In addition to a random permutation
      $\lperm \in \perms(\lparty)$, a random vertex $\lvertexgen \in \lparty$ is
    drawn and moved to rank $t$.}
    \label{fig:rankunmatchedmovetomonada}
  \end{subfigure}
  \begin{subfigure}[t]{.45\textwidth}
    \begin{flalign*}
      \bernoulli''_t \equiv{} &\kdo\ \{\ &\\
                                     & \sigma \leftarrow \uniform(\perms(\lparty)); \\
                                     & \lvertexgen \leftarrow \uniform(\lparty); \\
                                     & \klet\ R = \rank(\graph, \pi, \sigma); \\
                                     & \return\ (\matching(\lvertexgen) \in \vertices(R) \wedge \sigma(R(\matching(\lvertexgen)) \leq t) )\\
      \}
    \end{flalign*}
    \subcaption{Distribution describing the probability that the partner
      $\matching(\lvertexgen) \in \rparty$ of a random vertex
      $\lvertexgen \in \lparty$ is matched to a vertex of rank at most $t$.}
    \label{fig:rankunmatchedmovetomonadb}
  \end{subfigure}
  \caption{Two Bernoulli distributions used in the proof of \autoref{lem:ranktmatchedbound}}
  \label{fig:rankunmatchedmovetomonads}
\end{figure}

%% file: combinatorial.tex
\section{Lifting the Competitiveness to General Bi-Partite Graphs}

Until now, we have shown that \rankingprob\ satisfies the desired competitive
ratio for graphs with a perfect matching.
Also, until now, our formalisation closely follows BM's  proof.
However, in all previous graph-theoretic expositions of the correctness proof of this algorithm~\cite{GoelMehtaOnlineMatching2008,onlineMatchingSimple2008,karpOptimalAlgorithmOnline1990}, as opposed to linear programming-based  expositions~\cite{devanurRandomizedPrimalDualAnalysis2013,onlineMatchingEcon,vaziraniOnlineMatching2022}, the authors would stop at the current point, stating, or implicitly assuming, that it is obvious to see how the analysis of \rankingprob\ for bipartite graphs with perfect matchings extends to general bipartite graphs.
The central argument is as follows: it is easy to see that, for a fixed permutation of the offline vertices, if we remove a vertex from a bipartite graph that does not occur in a maximum matching of that graph, then \rank\ will compute a matching that is either one edge smaller or of the same size as the matching \rank\ would compute, given the original graph.

Indeed, BM, who are the authors who give the most detailed account of the graph-theoretic correctness proof of this algorithm, state, as a proof for this fact~\cite[Lemma~2]{onlineMatchingSimple2008}, that ``it is an easy structural observation''.
In a sense they are correct: in our example, illustrated in Fig.~\ref{fig:ranking}, if we remove $\rvertexb$, it is easy to see that \rank's output size will be only one less than on the original graph.
This is because all the matching edges will ``cascade'' down.
This is illustrated in Fig.~\ref{fig:online_g}, showing the blue edges being replaced with the red edges.
In this section we mainly formalise this argument.
We also formalise another easier, but no less crucial, graph-theoretic part of the proof by BM~\cite[Lemma 4]{onlineMatchingSimple2008}.
This lemma is used in the probabilistic part of the proof, as stated earlier.
In our formalisation we significantly simplified the proof.
Before we do so, however, we introduce some necessary background and notions related to paths.

\subsection{Alternating Paths, Augmenting Paths, and Berge's Lemma}

A list of vertices $[\vertexa,\vertexb,\dots,\vertexgen_n]$ is a path w.r.t.\ a
graph $\graph$ iff $\{\vertexgen_i, \vertexgen_{i+1}\}\in\graph$ for $1 \leq i < n$.
Note: a path $[\vertexa\vertexb\dots\vertexgen_n]$ is always a simple path as we only consider distinct lists.
A list of vertices $[\vertexa,\vertexb,\dots,\vertexgen_n]$ is an alternating path w.r.t.\ a set of edges $\edges$ iff for some $\edges'$ \begin{enumerate*}\item $\edges' = \edges$ or $\edges' \cap \edges=\emptyset$, \item $\{\vertexgen_i,\vertexgen_{i+1}\}\in\edges'$ holds for all even numbers $i$, where $1 \leq i < n$, and \item $\{\vertexgen_i,\vertexgen_{i+1}\}\not\in\edges'$ holds for all odd numbers $i$, where $1 \leq i \leq n$.\end{enumerate*}
We call a list of vertices $[\vertexa,\vertexb,\dots,\vertexgen_n]$ an augmenting path w.r.t. a matching $\matching$ iff $[\vertexa,\vertexb,\dots,\vertexgen_n]$ is an alternating path w.r.t. $\matching$ and $\vertexa,\vertexgen_n\not\in\vertices(\matching)$.
If $\matching$ is a matching w.r.t.\ a graph $\graph$, we call the path an augmenting path w.r.t. to the pair $\langle\graph,\matching\rangle$.
Also, for two sets $s$ and $t$, $s \oplus t$ denotes the symmetric difference of the two sets.

A central result in matching theory is Berge's lemma, which gives an
algorithmically useful characterisation of a maximum cardiniality matching.
\begin{mythm}[Berge's Lemma]
\label{thm:berge}
For a graph $\graph$, a matching $\matching$ is maximum w.r.t.\ $\graph$ iff there is not an augmenting path $\path$ w.r.t.\ $\langle \graph, \matching\rangle$.
\end{mythm}

We use a formalisation of the above concepts and Berge's Lemma by Abdulaziz et al.~\cite{DBLP:conf/mfcs/AbdulazizMN19}.
For completeness, the most important parts of this formalisation are demonstrated in Listing~\ref{isa:paths}.
Nonetheless, interested readers should refer to the original paper~\cite{DBLP:conf/mfcs/AbdulazizMN19}.

\subsection{\rank's Behaviour after Removing a Vertex}

Now that we have all the necessary machinery, we can discuss the formalisation of the correctness of \rankingprob\ for general bipartite graphs.
The central claim to show is stated in the following lemma, which is a restatement of Lemma~2 in BM's paper.
It states what happens to the result of \rank\ when a vertex is removed from the graph.

\newcommand{\graphsmall}{\ensuremath{\mathcal{H}}}
\begin{mylem}
\label{lem:lemma2}
Let $\graph$ be a bipartite graph w.r.t.\  the lists $\lperm$ and $\rperm$.
Consider a vertex $\rvertexgen\in\rperm$.
Let $\graphsmall$ be $G \setminus \{\rvertexgen\}$.
We have that either
$\rank(\graph,\rperm,\lperm) = \rank(\graphsmall,\rperm,\lperm)$ or
$\rank(\graph,\rperm,\lperm)\oplus\rank(\graphsmall,\rperm,\lperm)$ can be
ordered into an
alternating path w.r.t.\ $\rank(\graph,\rperm,\lperm)$ and w.r.t.\ $\rank(\graphsmall,\rperm,\lperm)$, and that path starts at $\lvertexgen$.
\end{mylem}

The above lemma was never proved by any of the previous expositions of the combinatorial argument for the algorithm's correctness.
BM's exposition is an exception, where there is at least a graphical example, showing what happens when we remove a vertex before running \rank.
A version of that graphical argument can be seen in Fig.~\ref{fig:ranking}.Fig.~\ref{fig:online_f} shows the matching computed by the algorithm on the original graph, and Fig.~\ref{fig:online_g} shows the difference in the computed matching if a vertex from the online side of the graph is removed. \footnote{The the lemma above is stated for an online vertex being removed, while in the formalisation an offline vertex is removed. This highlights an important concept in many of the proofs: the interchangeability of the offline and online vertices for fixed orders $\lperm$ and $\rperm$.}
As shown, when the vertex is removed, the matched edges ``cascade downwards'', where the original matching edges, shown in blue, are replaced with the red edges.
The statement of the lemma states that the symmetric difference between the two computed matchings is always an alternating path, w.r.t.\ both the old and the new matchings, if there is any difference.
When looking at the graphical illustration this is obvious.
However, when formalising that argument, many challenges manifest themselves.

The first challenge is the characterisation of the path that constitutes the difference between the two matchings.
This characterisation has to, among other things, make formal proofs by induction manageable.
To do so, we had to formulate this characterisation \emph{not} recursively on the given bipartite graph, i.e.\ the given bipartite graph should not change across different recursive calls.
Otherwise, proving anything about the path would involve a complicated induction on the given bipartite graph.

\input{shifts_to_listing}
To define that path, we first introduce a concept relating two vertices on the online side.
We state $\lvertexgen$ \emph{\shiftsto} $\lvertexgen'$
iff \begin{enumerate*}\item $\lvertexgen$ occurs before $\vertexgen'$ in the
  offline permutation $\lperm$, \item $\lvertexgen$ is matched to some
  $\rvertexgen$, \item $\lvertexgen'$ is not matched to any vertex that occurs
  before $\rvertexgen$ in $\rperm$, and \item any vertex
  $\lvertexgen''\in\neighb{\graph}{\rvertexgen}$ occurring between $\lvertexgen$
  and $\lvertexgen'$ in $\lperm$ is matched by \rank\ to a vertex occurring
  before $\rvertexgen$ in the arrival order $\rperm$.\end{enumerate*}
Intuitively, this means that, if $\vertexgen$ is removed from the graph, then $\vertexgen'$ would be matched to $\rvertexgen$ by \rank.
Our formalisation of this definition can found in Listing~\ref{isa:shifts_to}.
Note: the omitted arguments in the text, $\graph$, $\matching$, $\rperm$, and $\lperm$ are usually clear from the context.

Now that we are done with the definition of \shiftsto, we are ready to describe our characterisation of the path whose edges form the symmetric difference of the two matchings computed by \rank.
We characterise it using the following functions:

\begin{align*}
  \zig(\graph,\matching,\lvertexgen,\rperm,\lperm) &\equiv
                                                     \begin{cases}
                                                       \lvertexgen \concat \zag(\graph, \matching, \rvertexgen, \rperm, \lperm) & \text{if } \{\lvertexgen,\rvertexgen\} \in \matching \\
                                                       [\lvertexgen] & \text{otherwise}
                                                     \end{cases} \\
  \zag(\graph,\matching,\rvertexgen,\rperm,\lperm) &\equiv
                                                     \begin{cases}
                                                       \rvertexgen \concat \zag(\graph, \matching, \lvertexgen', \rperm, \lperm) & \text{if } \{\lvertexgen,\rvertexgen\} \in \matching \text{, for some $\lvertexgen$, and } \lvertexgen\ \shiftsto\ \lvertexgen' \\
                                                       [\rvertexgen] & \text{otherwise}
                                                     \end{cases}
\end{align*}

\input{zig_zag_listing}

As the names of the functions indicate, the path zig-zags between the online and the offline sides of the graph, going down the online ordering.
This is indicated in Fig.~\ref{fig:online_zig}. 
The formalisation of $\zig$-$\zag$ is given in Listing~\ref{isa:zig_zag}. 
Note that the formalisation has extra cases for when the second argument is not a matching: this is to ensure termination, which is not straightforward, as the definite descriptions are not well-defined in these cases.
The termination relation encodes the intuition that, while zig-zagging, the path also goes down the ordering of online vertices.
More formally, because this is a mutually recursive function, we have to provide an order that relates the argument passed to recursive calls of \zag\ from \zig\ and the other way around.
For evaluating $\zig(\graph,\matching,\lvertexgen,\rperm,\lperm)$, we need a
call to $\zag(\graph,\matching,\rvertexgen,\rperm,\lperm)$, in which case the
relation holds iff $\lvertexgen$ and $\rvertexgen$
satisfy \begin{enumerate*}\item
  $\{\lvertexgen,\rvertexgen\}\in\rank(\graph,\rperm,\lperm)$ and \item if there
  is $\lvertexgen'$, s.t.\ $\lvertexgen$ \shiftsto\ $\lvertexgen'$, then
  $\lperm(\lvertexgen) < \lperm(\lvertexgen')$\end{enumerate*}.
For evaluating $\zag(\graph,\matching,\rvertexgen,\rperm,\lperm)$, we need a
call to $\zig(\graph,\matching,\lvertexgen',\rperm,\lperm)$, in which case the
relation holds iff $\rvertexgen$ and $\lvertexgen'$
satisfy \begin{enumerate*}\item there is $\lvertexgen$ s.t.\
  $\{\lvertexgen,\rvertexgen\}\in\rank(\graph,\rperm,\lperm)$ and \item
  $\lvertexgen$ \shiftsto\ $\lvertexgen'$ and $\lperm(\lvertexgen) < \lperm(\lvertexgen')$\end{enumerate*}.

\newcommand{\rankmatching}{\textit{ranking-matching}}
\input{rank_matching_listing}
Another challenge for formalising the proof of Lemma~\ref{lem:lemma2} is devising a non-recursive characterisation of the properties of the matching computed by \rank, which would be enough for proving the lemma, yet more abstract than the actual specification of the algorithm.
This characterisation can be intuitively described as follows: $\matching$ is a \rankmatching\ w.r.t.\ $\graph$, $\lperm$, and
$\rperm$ iff \begin{enumerate*}\item $\graph$ is bipartite w.r.t.\ $\lperm$
  and $\rperm$, \item $\matching$ is a maximal matching w.r.t.\
  $\graph$, \item every vertex from $\rvertexgen \in \rperm$ is matched to the unmatched
  vertex from $\lperm$ at $\rvertexgen$'s arrival, to which it is connected,
  with the lowest rank in $\lperm$, and \item no vertex from $\lperm$ ``refuses'' to be matched.\end{enumerate*}
The formal specification is given in Listing~\ref{isa:rank_matching}.
It should be clear that the following properties hold for \rankmatching.
\begin{myprop}
\label{prop:rankingspec}
Let $\graph$ be a bipartite graph w.r.t.\ $\lperm$ and $\rperm$.
We have that \begin{enumerate*} \item $\rank(\graph,\rperm,\lperm)$ is a \rankmatching\ w.r.t.\ $\graph$, $\lperm$, and $\rperm$, \item if $\matching$ is a \rankmatching\ w.r.t.\ $\graph$, $\lperm$, and $\rperm$, then it is a \rankmatching\ w.r.t.\ $\graph$, $\rperm$, and $\lperm$, and \item if $\matching$ and $\matching'$ are both \rankmatching s w.r.t.\ $\graph$, $\lperm$, and $\rperm$, then $\matching=\matching'$.\end{enumerate*}
\end{myprop}


This specification of \rank\ makes our proofs about \rank\ much simpler,
as it allows us to gloss over many of the computational details of the algorithm.
In particulate, it allows us to avoid nested inductions, especially when using the I.H.\ of Lemma~\ref{lem:lemma2}.

Now that we have characterised the difference between the
  matchings computed by $\rank$ before and after removing a vertex, as well as the main properties satisfied by matchings computed by \rank, we are ready to present the proof that the competitiveness for bipartite graphs with perfect matchings lifts to general bipartite graphs.
\newcommand{\ihvar}[1]{\overline{#1}}
There are two main ideas to our proof.
The first one is that we show that the output of $\zig$, for some online vertex
$\rvertexgen$, which is matched to an offline vertex $\lvertexgen$, stays the
same when offline vertices are removed from the graph and the matching, if those
offline vertices are all ranked lower than $\lvertexgen$.
Graphically, this is clear.
For instance, in  Fig.~\ref{fig:online_zig}, if we remove the vertex $\lvertexa$ from the graph and the matching, the result of \zig\ applied to $\rvertexb$, w.r.t.\ to the modified graph and matching, will be the same as its output w.r.t.\ the old graph and matching.

\begin{mylem}
\label{lem:remove_verts_zig_zag}
Let $\graph$ be a bipartite graph w.r.t.\  $\lperm$ and $\rperm$.
Consider a vertex $\rvertexgen\in\rperm$, s.t.\ there is $\lvertexgen$, where $\{\lvertexgen,\rvertexgen\}\in\matching$.
Consider a set of vertices $\rparty'\subseteq\rperm$, s.t.\ for all
$\rvertexgen'\in\rparty'$ we have that $\rperm(\rvertexgen') < \rperm(\rvertexgen)$.\todo{assumption to strong, bipartite matching sufficient}
Let $\matching$ be a \rankmatching\ w.r.t.\ $\graph$, $\rperm$, and $\lperm$.
We have that $\zig(\graph,\matching,\rvertexgen,\lperm,\rperm)=\zig(\graph\setminus\rparty',\matching\setminus\rparty',\rvertexgen,\lperm,\rperm)$ and $\zag(\graph,\matching,\lvertexgen,\lperm,\rperm)=\zag(\graph\setminus\rparty',\matching\setminus\rparty',\lvertexgen,\rperm,\lperm)$.
\end{mylem}
We do not prove this lemma here: the proof depends on an involved case analysis of the behaviour of \shiftsto, and we describe below similar case analyses, which convey the difficulty of translating such obvious graphical arguments into proofs.
Interested readers, however, should refer to the accompanying formal proof.

The second idea is that we exploit the symmetry between the online and the offline vertices.
This is encoded in the following relationship between \zig\ and \zag.
\begin{mylem}
\label{lem:zigsymm}
Let $\graph$ be a bipartite graph w.r.t.\ $\lperm$ and $\rperm$.
Consider a vertex $\rvertexgen\in\rperm$.
Let $\graphsmall$ be $G \setminus \{\rvertexgen\}$.
Let $\matching$ be a \rankmatching\ w.r.t.\ $\graph$, $\rperm$, and $\lperm$, and $\matching'$ be a \rankmatching\ w.r.t.\ $\graphsmall$, $\lperm$, and $\rperm$.
Let $\lvertexgen$ be a vertex s.t.\ $\{\lvertexgen,\rvertexgen\}\in\matching$.
We have that $\zig(\graphsmall,\matching',\lvertexgen,\rperm,\lperm)=\zag(\graph,\matching,\lvertexgen,\lperm,\rperm)$.
\end{mylem}
Before we discuss the proof, we first show a graphical argument of why the lemma holds.
Fig.~\ref{fig:online_zig}~and~\ref{fig:online_zag} show an example of how \zig\ and \zag\ would return the same list of vertices if invoked on the same vertex once on the offline side, and another time on the online side.
In the first configuration, $\zag(\graph,\matching,\lvertexb,\lperm,\rperm)$ chooses $\rvertexc$, because in $\matching$, we have that $\lvertexb$ is matched to $\rvertexb$, and $\rvertexb$ \shiftsto\ $\rvertexc$.
Then the rest of the recursive calls proceed as shown in the figure.
When the online and offline sides are flipped, as shown in Fig.~\ref{fig:online_zag}, $\zig(\graphsmall,\matching',\lvertexb,\rperm,\lperm)$, where $\graphsmall$ denotes $\graph\setminus\{\rvertexb\}$, will also choose $\rvertexc$ because, this time, it will be matched to $\lvertexb$ in $\matching'$, which is a \rankmatching\ for $\graphsmall$.
As we will see in the proof, this graphical argument is much shorter than the corresponding textual proof, let alone the formal proof.
\begin{proof}
Our proof is by strong induction on the index of $\lvertexgen$.
Let all the variable names in the I.H.\ be barred, e.g.\ the graph is $\ihvar{\graph}$.
Our proof is done by case analysis.
We consider 3 cases: \begin{enumerate*}\item we have vertices $\rvertexgen'$, $\lvertexgen'$, s.t.\ $\{\lvertexgen,\rvertexgen'\}\in\matching'$ and $\{\rvertexgen',\lvertexgen'\}\in\matching$, \item we have a vertex $\rvertexgen'$, s.t.\ $\{\lvertexgen,\rvertexgen'\}\in\matching'$ and there is no $\lvertexgen'$ s.t.\ $\{\rvertexgen',\lvertexgen'\}\in\matching$, and \item there is no vertex $\rvertexgen'$, s.t.\ $\{\lvertexgen,\rvertexgen'\}\in\matching'$.\end{enumerate*}

We focus on the first case, as that is the one where we employ the I.H.
To apply the I.H., we use the following assignments of the quantified variables.
$\ihvar{\graph}\mapsto\graph\setminus\{\rvertexgen,\lvertexgen\}$, $\ihvar{\rperm}\mapsto\rperm$, $\ihvar{\lperm}\mapsto\lperm$, $\ihvar{\rvertexgen}\mapsto\rvertexgen'$, $\ihvar{\lvertexgen}\mapsto\lvertexgen'$, $\ihvar{\matching}\mapsto\matching\setminus\{\rvertexgen,\lvertexgen\}$, and $\ihvar{\matching'}\mapsto\matching'\setminus\{\lvertexgen,\rvertexgen'\}$.
From the I.H., we get $\zig(\ihvar\graphsmall,\ihvar\matching',\ihvar\lvertexgen,\ihvar\rperm,\ihvar\lperm)=\zag(\ihvar\graph,\ihvar\matching,\ihvar\rvertexgen,\ihvar\lperm,\ihvar\rperm)$.
This proof is then finished by Lemma~\ref{lem:remove_verts_zig_zag}.
\end{proof}

We are now ready to prove a lemma that immediately implies Lemma~\ref{lem:lemma2}.
\begin{mylem}
\label{lem:lemma2abs}
Let $\graph$ be a bipartite graph w.r.t.\  $\lperm$ and $\rperm$.
Consider a vertex $\rvertexgen\in\rperm$.
Let $\graphsmall$ be $G \setminus \{\rvertexgen\}$.
Let $\matching$ be a \rankmatching\ w.r.t.\ $\graph$, $\lperm$, and $\rperm$, and $\matching'$ be a \rankmatching\ w.r.t.\ $\graphsmall$, $\lperm$, and $\rperm$.
We have that $\matching\oplus\matching'=\zig(\graph,\matching,\rvertexgen,\lperm,\rperm)$\footnote{We abuse the notation: although $\zig(\graph,\matching,\rvertexgen,\lperm,\rperm)$ is the list of vertices in the path, we use it here to denote the edges in the path.} or $\matching=\matching'$.
\end{mylem}
\begin{proof}
Our proof is by strong induction on $\cardinality{\graph}$.
Again, let all the variable names in the I.H.\ be barred.
We consider two cases, either $\rvertexgen\notin\vertices(\matching)$ or $\rvertexgen\in\vertices(\matching)$.
In the former case, the lemma follows immediately, since \rank\ will compute the same matching.

For the second case, we instantiate the I.H.\ as follows:
$\ihvar{\graph}\mapsto\graph\setminus\{\rvertexgen\}$, $\ihvar\matching\mapsto\matching'$, $\ihvar{\matching'}\mapsto\matching\setminus\{\lvertexgen,\rvertexgen\}$, $\ihvar\rperm\mapsto\lperm$, $\ihvar{\lperm}\mapsto\rperm$, and $\ihvar{\rvertexgen}\mapsto\lvertexgen$, where $\lvertexgen$ is some vertex s.t.\ $\{\lvertexgen,\rvertexgen\}\in\matching$, which must exist since $\rvertexgen\in\vertices(\matching)$.\footnote{The instantiation of $\ihvar{\graphsmall}$ follows implicitly from the other instantiations.}
To show that the I.H.\ is usable in this case, we need to show that: \begin{enumerate*}\item $\ihvar{\matching}$ is a \rankmatching\ w.r.t.\ $\ihvar{\graph}$, $\ihvar{\rperm}$, and $\ihvar{\lperm}$, and \item $\ihvar{\matching'}$ is a \rankmatching\ w.r.t.\ $\ihvar{\graphsmall}$, $\ihvar{\rperm}$, and $\ihvar{\lperm}$\end{enumerate*}.
The first requirement follows from the assumption that $\matching'$ is \rankmatching\ w.r.t.\ $\graphsmall$, $\lperm$, and $\rperm$, and the fact that \rankmatching\ is commutative w.r.t.\ the left and right parties of the given graph.
The second requirement follows from a property of \rankmatching, which we do not
prove here, stating that for any $\matching$ that is a \rankmatching\ w.r.t.\
$\graph$, $\lperm$, and $\rperm$, and for any $\edgegen\in\matching$, $\matching-\{\edgegen\}$ is a \rankmatching\ w.r.t.\ $\graph\setminus\edgegen$, $\lperm$, and $\rperm$.

Then, from the I.H.\ and since we know that $\lvertexgen\in\vertices(\matching)$, we have that either \begin{enumerate*} \item $\ihvar\matching=\ihvar\matching'$ or \item $\ihvar{\matching}\oplus\ihvar{\matching'} = \zig(\ihvar{\graph},\ihvar{\matching},\ihvar{\rvertexgen},\ihvar{\lperm},\ihvar{\rperm})$.\end{enumerate*}
In the former case, we have that $\matching'=\matching\setminus\{\rvertexgen,\lvertexgen\}$, so $\lvertexgen$ was not matched to anything in the graph, after removing $\rvertexgen$.
This means that there is no $\rvertexgen'$ for $\lvertexgen$ s.t.\ $\rvertexgen$ \shiftsto\ $\rvertexgen'$, which means that $\zig(\graph,\matching,\rvertexgen,\lperm,\rperm) = [\rvertexgen,\lvertexgen]$.
From that, we have the lemma proved for this case, since $\matching\oplus\matching'=\{\lvertexgen,\rvertexgen\}$.

In the second case, we have that $\ihvar{\matching}\oplus\ihvar{\matching'} = \zig(\graph\setminus\{\rvertexgen\},\matching',\lvertexgen,\rperm,\lperm)$.
From Lemma~\ref{lem:zigsymm}, we have $\zig(\graph\setminus\{\lvertexgen\},\matching',\lvertexgen,\rperm,\lperm) = \zag(\graph,\matching,\lvertexgen,\lperm,\rperm)$.
From the definition of \zig\ and since $\{\rvertexgen,\lvertexgen\}\in\matching$, the lemma follows for this case.
\end{proof}

\begin{proof}[Proof of Lemma~\ref{lem:lemma2}]
Lemma~\ref{lem:lemma2} follows immediately from Lemma~\ref{lem:lemma2abs} lemma and from Proposition~\ref{prop:rankingspec}.
\end{proof}

\subsection{Finishing the Proof}

The next step in our proof is to generalise the previous analysis to address the case when the removed vertex is from the offline side of the graph.
Although this is not considered by any of the previous expositions, this generalisation is crucial for proving the competitive ratio for general bipartite graphs, i.e.\ graphs that do not have a perfect matching.
\begin{mylem}
\label{lem:lemma2absoffline}
Let $\graph$ be a bipartite graph w.r.t.\ $\lperm$ and $\rperm$.
Consider a vertex $\lvertexgen\in\lperm$.
Let $\graphsmall$ be $G \setminus \{\lvertexgen\}$.
Let $\matching$ be a \rankmatching\ w.r.t.\ $\graph$, $\lperm$, and $\rperm$, and $\matching'$ be a \rankmatching\ w.r.t.\ $\graphsmall$, $\lperm$, and $\rperm$.
We have that $\matching\oplus\matching'=\zig(\graph,\matching,\lvertexgen,\rperm,\lperm)$ or $\matching=\matching'$.
\end{mylem}
The proof of this lemma is very similar to that of Lemma~\ref{lem:lemma2}.
However, we are able to reuse all our lemmas that exploit the symmetry of the offline and online sides of the graphs, so there is not much redundancy in our proofs.

Until now, we have primarily focused on the \emph{structural} difference between matchings computed by \rank\ before and after removing a vertex from the original graph.
The next step in the proof is to use that to reason about the competitiveness of \rank\ for general bipartite graphs.
The first step is proving the following lemma.
\begin{mylem}
\label{lem:removevtxcardmatching}
Let $\graph$ be a bipartite graph w.r.t.\  $\lperm$ and $\rperm$.
Consider a vertex $x$.
Let $\graphsmall$ be $G \setminus \{x\}$.
Let $\matching$ be a \rankmatching\ w.r.t.\ $\graph$, $\lperm$, and $\rperm$, and $\matching'$ be a \rankmatching\ w.r.t.\ $\graphsmall$, $\lperm$, and $\rperm$.
We have that $\cardinality{\matching'}\leq\cardinality{\matching}$.
\end{mylem}
\begin{proof}
Our proof is by case analysis.
The first case is when $x\notin\vertices(\matching)$.
In this case we will have that $\matching = \matching'$, which finishes our proof.

The second case is when $x\in\vertices(\matching)$.
In this case, we have two sub-cases: either $x\in\rperm$ or $x\in\lperm$.
We only describe the first case here and the second is symmetric.
Our proof is by contradiction, i.e.\ assuming $\cardinality{\matching'}>\cardinality{\matching}$.
From Lemma~\ref{lem:lemma2abs}, we have that $\matching\oplus\matching'=\zig(\graph,\matching,\rvertexgen,\lperm,\rperm)$.
Also note that, from Berge's lemma, we will have that a subsequence of $\zig(\graph,\matching,\rvertexgen,\lperm,\rperm)$ is an augmenting path w.r.t.\ $\langle\graph,\matching\rangle$.
We know from the definition of an augmenting path that, both, its first and last vertices are not in the matching it augments.
Accordingly, we have that the first and last vertices of that subsequence of
$\zig(\graph,\matching,\rvertexgen,\lperm,\rperm)$ are not in $\matching$.
This is a contradiction, because all vertices in $\zig(\graph, \matching, \rvertexgen, \lperm, \rperm)$, except possibly the last one, are in $\vertices(\matching)$.
\end{proof}

\newcommand{\removeunmatched}{\textit{make-perfect}}
\input{make_perfect_listing}
Lastly, we show that, given a bipartite graph $\graph$ and a maximum cardinality matching $\matching$ for that graph, we can recursively remove the vertices that do not occur in $\matching$.
To do that we define a recursive function, \removeunmatched, to remove these vertices and then prove the following lemma by computation induction, using the computation induction principle corresponding to \removeunmatched.
Listing~\ref{isa:make_perfect} shows the formalisation of that function.
\begin{mylem}
\label{lem:removevtxcardmatching}
Let $\graph$ be a bipartite graph w.r.t.\  $\lperm$ and $\rperm$.
Let $\matching$ be a \rankmatching\ w.r.t.\ $\graph$, $\lperm$, and $\rperm$, and $\matching'$ be a \rankmatching\ w.r.t.\ $\removeunmatched(\graph,\matching)$, $\lperm$, and $\rperm$.
We have that $\cardinality{\matching'}\leq\cardinality{\matching}$.
\end{mylem}

This last lemma leads to the final theorem below.
\begin{mythm}
\label{thm:rankcorrectgeneral}
Let $\graph$ be a bipartite graph w.r.t.\  $\lperm$ and $\rperm$.
Let $\matching$ be a maximum cardinality matching for $\graph$.
We have that $1-(1 - \frac{1}{\cardinality{\matching}+1})^\cardinality{\matching}\leq\mathbb{E}_{R \sim \rankingprob(\graph, \rperm)}[\cardinality{R}]/\cardinality{\matching}$.
\end{mythm}
\begin{proof}
This follows immediately from Lemma~\ref{lem:removevtxcardmatching}, Theorem~\ref{thm:rankcorrectperfect}, and the fact that the size of a maximum cardinality matching for $\removeunmatched(\graph,\matching)$ is the same as the size of $\matching$, if $\matching$ is a maximum cardinality matching for $\graph$.
\end{proof}

\subsection{Proving Lemma~\ref{lem:rankunmatchedmoveto}}
\input{rank_unmatched_moveto_fig}

Until now we have not discussed how we formalised Lemma~\ref{lem:rankunmatchedmoveto} -- we believe it better fits here as its proof is a combinatorial argument.
Graphically, Fig.~\ref{fig:rankunmatchedmoveto} shows some instances of Lemma~\ref{lem:rankunmatchedmoveto} for $\lvertexgen = \lvertexc$ and $\matching(\lvertexc) = \rvertexa$.
No matter where $\lvertexc$ is put, $\rvertexa$ is always matched to a vertex of rank at most $3$.
BM prove this Lemma by stating that the difference, if any, between the matchings computed by $\rank$ before and after moving the offline vertex is also an alternating path, where the ranks of the offline vertices traversed by that path increase. 
Again, like other combinatorial parts of the analysis, graphically this is clearly evident: Fig.~\ref{fig:rankunmatchedmovetoupmatched} shows the difference between $\rank(\graph, \rperm, \lperm)$ and $\rank(\graph, \rperm, \lperm[\lvertexc \mapsto 1])$.
The blue edge was removed from the original matching, and the two red edges are added instead.
The three edges form an alternating path w.r.t.\ to the original matching.

However, to formalise this argument would be as difficult as for Lemma~\ref{lem:lemma2}.
Indeed, we found out that there is no reason to construct the entire difference between the two matchings just to reason about the rank of the vertex $\lvertexgen_i$ to which $\rvertexgen$ is matched in $\rank(\graph, \rperm, \lperm[\lvertexgen \mapsto i])$.
With this approach, the lemma follows almost immediately from the specification $\rankmatching$.
Hence, the formal proof is much shorter than BM's approach.

\section{The Competitive Ratio in the Limit}
\input{rankcorrectlim_listing}

BM claim that the competitive ratio tends to $1-1/e$ if the matching's size tends to infinity.
The main complication of showing that is to show that the competitive ratio converges, which they do not address at all.
We formalised the following.
\newcommand{\graphfamily}{\ensuremath{\Gamma}}
\newcommand{\compratio}{\ensuremath{\mathcal{Q}}}

\begin{mythm}
\label{thm:rankcorrectlim}
Let $\matching_n$ denote $\{\{(0,k),(1,k)\}\mid 1\leq k \leq n\}$.
Let $\graphfamily_n$ denote graphs in the power set of $\{\{(0,k),(1,l)\}\mid 1\leq k,l \leq 2n\}$ and that have $\matching_n$ as a maximum cardinality matching.
Let $\rperm_n$ denote $\perms(\{(1,k)\mid1\leq k \leq 2n\})$.
If $\compratio_n$ converges, then $\compratio_n$ tends to $1-1/e$ as $n$ tends to $\infty$, where $\compratio_n$ denotes $\min_{(\graph, \rperm)\in\graphfamily_n\times\rperm_n}\mathbb{E}_{R \sim \rankingprob(\graph, \rperm)}[\cardinality{R}]/\cardinality{\matching_n}$.
\end{mythm}

We only prove the limit for a specific set of bipartite graphs, namely, $\graphfamily_n$.
We conjecture that $\graphfamily_n$ is isomorphic to the set of all bipartite graphs with maximum cardinality matchings of size $n$.
Despite it being trivial, it was impressive that the part of the proof of this lemma which pertains to arithmetic manipulation was almost completely automated using Eberl's tool~\cite{manuelLimitsISSAC}.
The other part of the proof was to show that $\graphfamily_n$ is finite, which was tedious.

The more interesting part would be to show that $\compratio_n$ converges.
In BM, they do not prove that, yet they do not have it as an assumption in their theorem statement.
One way to show that this assumption holds is to use the theorem by KVV showing that no online algorithm for bipartite matching has a better competitive ratio that $1-1/e$.
However, formalising that theorem is beyond the scope of our project.


%% file: shifts_to_listing.tex
\begin{figure*}[t]
\begin{IsabelleSnippet}[label=isa:shifts_to]{Formalising \shiftsto\ in Isabelle/HOL}
definition "shifts_to G M u v v' \<pi> \<sigma> \<equiv>
 u \<in> set \<pi> \<and> v' \<in> set \<sigma> \<and> index \<sigma> v < index \<sigma> v' \<and> {u,v'} \<in> G
 \<and> (\<nexists>u'. index \<pi> u' < index \<pi> u \<and> {u',v'} \<in> M) \<and>
 (\<forall>v''. (index \<sigma> v < index \<sigma> v'' \<and> index \<sigma> v'' < index \<sigma> v')
 \<longrightarrow> ({u,v''} \<notin> G \<or> (\<exists>u'. index \<pi> u' < index \<pi> u \<and> {u',v''} \<in> M)))
\end{IsabelleSnippet}
\end{figure*}

%% file: zig_zag_listing.tex
\begin{figure*}[t]
\begin{IsabelleSnippet}[label=isa:zig_zag]{Formalising zig-zag and their termination relation in Isabelle/HOL.}
function zig :: "'a graph \<Rightarrow> 'a graph \<Rightarrow> 'a \<Rightarrow> 'a list \<Rightarrow> 'a list \<Rightarrow> 'a list"
and zag :: "'a graph \<Rightarrow> 'a graph \<Rightarrow> 'a \<Rightarrow> 'a list \<Rightarrow> 'a list \<Rightarrow> 'a list" where
  proper_zig: "zig G M v \<pi> \<sigma> = v # (
                    if \<exists>u. {u,v} \<in> M 
                    then zag G M (THE u. {u,v} \<in> M) \<pi> \<sigma>
                    else [])" if "matching M"
| no_matching_zig: "zig _ M v _ _ = [v]" if "\<not>matching M"

| proper_zag: "zag G M u \<pi> \<sigma> =  u # (if \<exists>v. {u,v} \<in> M
                      then 
                      (let v = THE v. {u,v} \<in> M in (
                        if \<exists>v'. shifts_to G M u v v' \<pi> \<sigma>
                        then zig G M (THE v'. shifts_to G M u v v' \<pi> \<sigma>) \<pi> \<sigma>
                        else [])
                      )
                      else []
                    )" if "matching M"
| no_matching_zag: "zag _ M v _ _ = [v]" if "\<not>matching M"

\end{IsabelleSnippet}
\end{figure*}


%% file: rank_matching_listing.tex
\begin{figure*}[t]
\begin{IsabelleSnippet}[label=isa:rank_matching]{Formalising the specification of \rank's output in Isabelle/HOL.}
definition ranking_matching :: "'a graph \<Rightarrow> 'a graph \<Rightarrow> 'a list \<Rightarrow> 'a list \<Rightarrow> bool" where
  "ranking_matching G M \<pi> \<sigma> \<equiv> graph_matching G M \<and>
    bipartite G (set \<pi>) (set \<sigma>) \<and> maximal_matching G M \<and>
    (\<forall>u v v'. ({u,v}\<in>M \<and> {u,v'}\<in>G \<and> index \<sigma> v' < index \<sigma> v) \<longrightarrow>
      (\<exists>u'. {u',v'}\<in>M \<and> index \<pi> u' < index \<pi> u)) \<and>
    (\<forall>u v u'. ({u,v}\<in>M \<and> {u',v}\<in>G \<and> index \<pi> u' < index \<pi> u) \<longrightarrow>
      (\<exists>v'. {u',v'}\<in>M \<and> index \<sigma> v' < index \<sigma> v))"
\end{IsabelleSnippet}
\end{figure*}

%% file: make_perfect_listing.tex
\begin{figure*}[t]
\begin{IsabelleSnippet}[label=isa:make_perfect]{Formalising the specification of \removeunmatched's output in Isabelle/HOL.}
function make_perfect_matching :: "'a graph \<Rightarrow> 'a graph \<Rightarrow> 'a graph" where
  "make_perfect_matching G M = (
    if (\<exists>x. x \<in> Vs G \<and> x \<notin> Vs M)
    then make_perfect_matching (G \<setminus> {SOME x. x \<in> Vs G \<and> x \<notin> Vs M}) M
    else G
  )
  " if "finite G"
| "make_perfect_matching G M = G" if "infinite G"
\end{IsabelleSnippet}
\end{figure*}

%% file: rank_unmatched_moveto_fig.tex
\begin{figure}[t]
  \centering
  \begin{subfigure}[t]{.24\textwidth}
    \begin{tikzpicture}[xscale=-1]
      \input{rank_unmatched_moveto}
    \end{tikzpicture}
    \subcaption{}
    \label{fig:rankunmatchedmovetoinit}
  \end{subfigure}
  \begin{subfigure}[t]{.24\textwidth}
    \begin{tikzpicture}[xscale=-1]
      \input{rank_unmatched_moveto_down}
    \end{tikzpicture}
    \subcaption{}
    \label{fig:rankunmatchedmovetodown}
  \end{subfigure}
  \begin{subfigure}[t]{.24\textwidth}
    \begin{tikzpicture}[xscale=-1]
      \input{rank_unmatched_moveto_up_unmatched}
    \end{tikzpicture}
    \subcaption{}
    \label{fig:rankunmatchedmovetoupunmatched}
  \end{subfigure}
  \begin{subfigure}[t]{.24\textwidth}
    \begin{tikzpicture}[xscale=-1]
      \input{rank_unmatched_moveto_up_matched}
    \end{tikzpicture}
    \subcaption{}
    \label{fig:rankunmatchedmovetoupmatched}
  \end{subfigure}
  \caption{Illustrating \autoref{lem:rankunmatchedmoveto}, where
    $\lvertexgen = \lvertexc$, and $\matching(\lvertexc) = \rvertexa$.
    Initially (\ref{fig:rankunmatchedmovetoinit}), $\lvertexc$ is unmatched. Moving it further
    down in the ranking (\ref{fig:rankunmatchedmovetodown}) does not change the
    partner of $\rvertexa$. Moving $\lvertexc$ up in the ranking can either
    (\ref{fig:rankunmatchedmovetoupunmatched}) also leave $\rvertexa$ untouched,
  or (\ref{fig:rankunmatchedmovetoupmatched}) change the partner of $\rvertexa$.}
  \label{fig:rankunmatchedmoveto}
\end{figure}


%% file: rank_unmatched_moveto.tex
{\node (ghost) at (0,0) [varnode] {\textcolor{white}{$\vertexa$}} ;}
{\node (ghost2) at (0,-3.5) [varnode] {\textcolor{white}{$\vertexa$}} ;}
{\node (b1) at (2,0) [varnode] {\textcolor{black}{$\lvertexa$}} ;}
{\node (b2) at (2,-1) [varnode] {\textcolor{black}{$\lvertexb$}} ;}
{\node (b3) at (2,-2) [varnode] {\textcolor{black}{$\lvertexc$}} ;}
{\node (b4) at (2,-3) [varnode] {\textcolor{black}{$\lvertexd$}} ;}

\input{online_perm_small}

\draw [green,-,matchededge] (a1) -- (b1);
\draw [green,-,matchededge] (a3) -- (b4);
\draw [green,-,matchededge] (a2) -- (b2);

\draw [black,-,unmatchededge] (a4) -- (b1);
\draw [black,-,unmatchededge] (a1) -- (b3);


%% file: rank_unmatched_moveto_down.tex
{\node (ghost) at (0,0) [varnode] {\textcolor{white}{$\vertexa$}} ;}
{\node (b1) at (2,0) [varnode] {\textcolor{black}{$\lvertexa$}} ;}
{\node (b2) at (2,-1) [varnode] {\textcolor{black}{$\lvertexb$}} ;}
{\node (b4) at (2,-3) [varnode] {\textcolor{black}{$\lvertexd$}} ;}
{\node (b3) at (2,-3.5) [varnode] {\textcolor{red}{$\lvertexc$}} ;}

\input{online_perm_small}

\draw [black,-,unmatchededge] (a4) -- (b1);

\draw [green,-,matchededge] (a1) -- (b1);
\draw [green,-,matchededge] (a3) -- (b4);
\draw [green,-,matchededge] (a2) -- (b2);

\draw [black,-,unmatchededge] (a1) -- (b3);


%% file: rank_unmatched_moveto_up_unmatched.tex
{\node (ghost) at (0,0) [varnode] {\textcolor{white}{$\vertexa$}} ;}
{\node (ghost2) at (0,-3.5) [varnode] {\textcolor{white}{$\vertexa$}} ;}
{\node (b1) at (2,0) [varnode] {\textcolor{black}{$\lvertexa$}} ;}
{\node (b3) at (2,-0.5) [varnode] {\textcolor{red}{$\lvertexc$}} ;}
{\node (b2) at (2,-1) [varnode] {\textcolor{black}{$\lvertexb$}} ;}
{\node (b4) at (2,-3) [varnode] {\textcolor{black}{$\lvertexd$}} ;}

\input{online_perm_small}

\draw [green,-,matchededge] (a1) -- (b1);
\draw [green,-,matchededge] (a3) -- (b4);
\draw [green,-,matchededge] (a2) -- (b2);

\draw [black,-,unmatchededge] (a4) -- (b1);
\draw [black,-,unmatchededge] (a1) -- (b3);


%% file: rank_unmatched_moveto_up_matched.tex
{\node (ghost) at (0,0) [varnode] {\textcolor{white}{$\vertexa$}} ;}
{\node (ghost2) at (0,-3.5) [varnode] {\textcolor{white}{$\vertexa$}} ;}
{\node (b3) at (2,0.5) [varnode] {\textcolor{red}{$\lvertexc$}} ;}
{\node (b1) at (2,0) [varnode] {\textcolor{black}{$\lvertexa$}} ;}
{\node (b2) at (2,-2) [varnode] {\textcolor{black}{$\lvertexb$}} ;}
{\node (b4) at (2,-3) [varnode] {\textcolor{black}{$\lvertexd$}} ;}

\input{online_perm_small}

\draw [red,-,matchededge] (a1) -- (b3);
\draw [green,-,matchededge] (a3) -- (b4);
\draw [green,-,matchededge] (a2) -- (b2);
\draw [red,-,matchededge] (a4) -- (b1);

\draw [blue,-,unmatchededge] (a1) -- (b1);


%% file: rankcorrectlim_listing.tex
\begin{figure*}[t]
\begin{IsabelleSnippet}[label=isa:rankcorrectlim]{The formalisation of Theorem~\ref{thm:rankcorrectlim}}
abbreviation matching_instance_nat :: "nat \<Rightarrow> (nat \<times> nat) graph" where
  "matching_instance_nat n \<equiv> {{(0,k),(Suc 0,k)} |k. k < n}"

definition ranking_instances_nat :: "nat \<Rightarrow> (nat \<times> nat) graph set" where
  "ranking_instances_nat n \<equiv> {G. max_card_matching G (matching_instance_nat n) \<and>
     finite G \<and> G \<subseteq> {{(0,k),(Suc 0,l)} |k l. k < 2*n \<and> l < 2*n}}"

definition arrival_orders :: "(nat \<times> nat) graph \<Rightarrow> (nat \<times> nat) list set" where
  "arrival_orders G \<equiv> permutations_of_set {(Suc 0,l) |l. \<exists>k. {(0,k),(Suc 0,l)} \<in> G}"

definition offline_vertices :: "(nat \<times> nat) graph \<Rightarrow> (nat \<times> nat) set" where
  "offline_vertices G \<equiv> {(0, k) |k. \<exists>l. {(0, k),(Suc 0, l)} \<in> G}"

definition comp_ratio_nat where
  "comp_ratio_nat n \<equiv>
     Min {Min {measure_pmf.expectation 
                (wf_ranking.ranking_prob G \<pi> (offline_vertices G)) card
                   / card (matching_instance_nat n)
                  |\<pi>. \<pi> \<in> arrival_orders G}
                     | G. G \<in> ranking_instances_nat n}"

theorem comp_ratio_limit':
  assumes "convergent comp_ratio_nat"
  shows "1 - exp(-1) \<le> (lim comp_ratio_nat)"
\end{IsabelleSnippet}
\end{figure*}

%% file: conclusion.tex
\section{Discussion}

KVV's paper on online bipartite matching was a seminal result in the theory of online algorithms and matching.
Its interesting theoretical properties, together with the emergence of online matching markets have inspired a lot of generalisations to other settings, e.g.\ for weighted vertices~\cite{onlineVertexWeightedMatching}, online bipartite b-matching~\cite{onlinebmatching}, the AdWords market~\cite{AdWords2007}, which models the multi-billion dollars industry online advertising industry, and general graphs~\cite{onlineMatchingGeneralGraphs}, which models applications like ride-sharing.
All of this means an improved understanding of the theory of online-matching, and especially RANKING, is of great interest.

Indeed, as stated earlier, multiple authors studied the analysis of RANKING.
We mention here the most relevant five approaches:
\begin{enumerate*}
    \item Goel and A. Mehta~\cite{GoelMehtaOnlineMatching2008}, tried to simplify the proof and fill in a ``hole'' in KVV's original proof, in particular in the proof of Lemma~6 in KVV's original paper, 
    \item Birnbaum and C. Mathieu~\cite{onlineMatchingSimple2008} also provided a simple, primarily combinatorial, proof for RANKING, 
    \item Devanur, Jain, and Kleinberg~\cite{devanurRandomizedPrimalDualAnalysis2013} whose main contribution was to model the algorithm as a primal-dual algorithm, in an attempt to unify the approaches for analysing the unweighted, vertex-weighted, and the AdWords problem,
    \item Eden, Feldman, Fiat, and Segal~\cite{onlineMatchingEcon}, who tried to simplify the proof by using approaches from theory of economics, and finally
    \item Vazirani~\cite{vaziraniOnlineMatching2022}, who tried to simplify the proof of RANKING, in an attempt to use RANKING, or a generalisation of it, to solve AdWords.
\end{enumerate*}
However, despite all of these attempts, the proof of RANKING's correctness is still considered difficult to understand, e.g.\ Vazirani's latest trial to
generalize it had a critical non-obvious flaw in the combinatorial part of the analysis~\cite{vaziraniOnlineMatching2022}, which took months of reviewing to find out.

We believe this formalisation serves two purposes.
First, it is yet another attempt to further the understanding of this algorithm's analysis.
From that perspective, our work achieved two things. \begin{enumerate*}\item It further clarified the complexity of the combinatorial argument underlying the analysis of this algorithm by providing a detailed proof for how one could generalise the competitiveness of the algorithm from bipartite graphs with perfect matchings to general bipartite graphs.
We note that this part of the analysis is analogous to the ``no-surpassing
property'' in Vazirani's work~\cite{vaziraniOnlineMatching2022}, which is where his attempt to generalise RANKING to AdWords fell apart, further confirming our findings regarding the complexity of this part of the analysis.
\item We significantly simplified the analysis of the consequences of changing the ranking of an offline vertex.
\end{enumerate*}

Another outcome of this project is interesting from a formalisation perspective.
It further confirmed the previously reported observation that it is particularly hard to formalise graphical or geometric arguments and concepts.
E.g.\ verbally, let alone formally, encoding the intuition behind \shiftsto, which is a primarily graphical concept, is extremely cumbersome.
We hypothesise that this is an inherent complexity in graphical concepts and arguments which manifests itself when the graphical argument is put into prose.




One point which we believe would particularly benefit from further study is that of modelling online computation.
In its full generality, online computation is computation where the algorithm has access only to parts of the input, which arrive serially, but not the whole input.
The way we model our algorithm is ad-hoc and does not capture that essence of online computation in its full generality.
It remains an interesting question how can one model online computation, more generally.
In addition to the theoretical interest, a satisfactory answer to that question is essential if one is to show that the competitive ratio of RANKING is optimal for online algorithms, which is a main result of KVV.



%% file: supp.tex
\section*{Appendix: Isabelle/HOL Listings}
\input{graph_matching_listing}
\input{do_programs_listing}

\input{paths_listings}
\input{ranktmatchedbound_listing}
\input{remove_verts_zig_zag_listing}
\input{zigsymm_listing}
\input{lemma2abs_listing}
\input{removevtxcardmatching_listing}
\input{rankcorrectgeneral_listing}

%% file: graph_matching_listing.tex
\begin{figure*}[h]
\begin{IsabelleSnippet}[label=isa:graph_matching]{The formalisation of graphs and matching we use in Isabelle/HOL}
locale graph_defs =
  fixes E :: "'a set set"

abbreviation "graph_invar E \<equiv> (\<forall>e\<in>E. \<exists>u v. e = {u, v} \<and> u \<noteq> v) \<and> finite (Vs E)"

locale graph_abs =
  graph_defs +
  assumes graph: "graph_invar E"

definition matching where
  "matching M \<longleftrightarrow> (\<forall>e1 \<in> M. \<forall>e2 \<in> M. e1 \<noteq> e2 \<longrightarrow> e1 \<inter> e2 = {})"

\end{IsabelleSnippet}
\end{figure*}

%% file: do_programs_listing.tex
\begin{figure*}[t]
\begin{IsabelleSnippet}[label=isa:do_programs]{Formalising the different distributions we need in our proof.}
abbreviation rank_matched :: "nat \<Rightarrow> bool pmf" where
  "rank_matched t \<equiv>
    do {
      \<sigma> \<leftarrow> pmf_of_set (permutations_of_set V);
      let R = online_match G \<pi> \<sigma>;
      return_pmf (\<sigma> ! t \<in> Vs R)
    }"

definition matched_before :: "nat \<Rightarrow> bool pmf" where
  "matched_before t \<equiv>
    do {
      \<sigma> \<leftarrow> pmf_of_set (permutations_of_set V);
      v \<leftarrow> pmf_of_set V;
      let R = online_match G \<pi> \<sigma>; 
      let u = (THE u. {u,v} \<in> M);
      return_pmf (u \<in> Vs R \<and> index \<sigma> (THE v. {u,v} \<in> R) \<le> t)
    }"

lemma matched_before_uniform_u: "matched_before t = do
    {
      \<sigma> \<leftarrow> pmf_of_set (permutations_of_set V);
      u \<leftarrow> pmf_of_set (set \<pi>);
      let R = online_match G \<pi> \<sigma>;
      return_pmf (u \<in> Vs R \<and> index \<sigma> (THE v. {u,v} \<in> R) \<le> t)
    }"

abbreviation "matched_before_t_set t \<equiv> 
  do {
    \<sigma> \<leftarrow> pmf_of_set (permutations_of_set V);
    let R = online_match G \<pi> \<sigma>;
    return_pmf {u \<in> set \<pi>. u \<in> Vs R \<and> index \<sigma> (THE v. {u,v} \<in> R) \<le> t}
  }"
\end{IsabelleSnippet}
\caption*{$\bernoulli_t$ is formalised as \texttt{rank\_matched}, $\bernoulli_t''$ is formalised as \texttt{matched\_before}, \texttt{matched\_before\_uniform\_u} is the formal statement showing that the distribution of a randomly chosen online vertex is matched to an offline vertex of rank at most $t$ is the same as $\bernoulli_t''$, and $\bernoulli_t^*$ is formalised as \texttt{matched\_before\_t\_set}.}
\end{figure*}

%% file: paths_listings.tex
\begin{figure*}[h]
\begin{IsabelleSnippet}[label=isa:paths]{The definitions of paths and augmenting paths and Berge's lemma as formalised in Isabelle/HOL.}
context fixes G :: "'a set set" begin
inductive path where
  path0: "path []" |
  path1: "v \<in> Vs G \<Longrightarrow> path [v]" |
  path2: "{v,v'} \<in> G \<Longrightarrow> path (v'#vs) \<Longrightarrow> path (v#v'#vs)"
end

inductive alt_list where
"alt_list P1 P2 []" |
"P1 x \<Longrightarrow> alt_list P2 P1 l \<Longrightarrow> alt_list P1 P2 (x#l)"

definition augmenting_path where
  "augmenting_path M p \<equiv>
    alt_list (\<lambda>e. e \<notin> M) (\<lambda>e. e \<in> M) (edges_of_path p)
    \<and> (length p \<ge> 2) \<and> hd p \<notin> Vs M \<and> last p \<notin> Vs M"

abbreviation "augpath E M p \<equiv> path E p \<and> distinct p \<and> augmenting_path M p"

lemma Berge_1:
  assumes finite: "finite M" "finite M'" and
    matchings: "matching M" "matching M'" and
    lt_matching: "card M < card M'" and
    doubleton_neq_edges: "\<forall>e\<in>(M \<oplus> M').\<exists>u v. e = {u,v} \<and> u \<noteq> v" "\<forall>e\<in>M. \<exists>u v. e = {u, v} \<and> u \<noteq> v"
  shows "\<exists>p. augmenting_path M p \<and> path (M \<oplus> M') p \<and> distinct p"
\end{IsabelleSnippet}
\end{figure*}

%% file: ranktmatchedbound_listing.tex
\begin{figure*}[t]
\begin{IsabelleSnippet}[label=isa:ranktmatchedbound]{Formal statement of Lemma~\ref{lem:ranktmatchedbound}.}
lemma rank_t_unmatched_prob_bound: 
  "t < card V \<Longrightarrow> 
     1 - measure_pmf.prob (rank_matched t) {True} \<le> 
        1 / (card V) * (\<Sum>s\<le>t. measure_pmf.prob (rank_matched s) {True})"
\end{IsabelleSnippet}
\end{figure*}

%% file: remove_verts_zig_zag_listing.tex
\begin{figure*}[t]
\begin{IsabelleSnippet}[label=isa:ranktmatchedbound]{Formal statement of Lemma~\ref{lem:remove_verts_zig_zag}.}
lemma
  assumes "X \<subseteq> set \<pi>"
  assumes "bipartite M (set \<pi>) (set \<sigma>)"
  assumes "matching M"
  shows 
   remove_online_vertices_zig_zig_eq: 
     "v \<in> set \<sigma> \<Longrightarrow> 
        \<forall>x \<in> X. ((\<exists>v'. {x,v'} \<in> M) \<longrightarrow> index \<sigma> (THE v'. {x,v'} \<in> M) < index \<sigma> v) \<Longrightarrow> 
                  zig (G \<setminus> X) (M \<setminus> X) v \<pi> \<sigma> = zig G M v \<pi> \<sigma>" and
   remove_online_vertices_zag_zag_eq:
     "u \<in> set \<pi> \<Longrightarrow>
     ((\<exists>v. {u,v} \<in> M \<Longrightarrow> 
           \<forall>x \<in> X. ((\<exists>v. {x,v} \<in> M) \<longrightarrow> 
                  index \<sigma> (THE v. {x,v} \<in> M) < index \<sigma> (THE v. {u,v} \<in> M)))) \<Longrightarrow>
         zag (G \<setminus> X) (M \<setminus> X) u \<pi> \<sigma> = zag G M u \<pi> \<sigma>"
\end{IsabelleSnippet}
\end{figure*}

%% file: zigsymm_listing.tex
\begin{figure*}[t]
\begin{IsabelleSnippet}[label=isa:zigsymm]{Formal statement of Lemma~\ref{lem:zigsymm}.}
lemma\<^marker>\<open>tag important\<close> ranking_matching_zig_zag_eq:
  assumes "{u,x} \<in> M"
  assumes "x \<in> set \<sigma>"
  assumes "ranking_matching G M \<pi> \<sigma>"
  assumes "ranking_matching (G \<setminus> {x}) M' \<sigma> \<pi>"
  shows "zig (G \<setminus> {x}) M' u \<sigma> \<pi> = zag G M u \<pi> \<sigma>"
\end{IsabelleSnippet}
\end{figure*}

%% file: lemma2abs_listing.tex
\begin{figure*}[t]
\begin{IsabelleSnippet}[label=isa:lemma2abs]{Formal statement of Lemma~\ref{lem:lemma2abs}.}
lemma remove_offline_vertex_diff_is_zig:
  assumes "ranking_matching G M \<pi> \<sigma>"
  assumes "ranking_matching (G \<setminus> {x}) M' \<pi> \<sigma>"
  assumes "x \<in> set \<sigma>"
  shows "M \<oplus> M' = set (edges_of_path (zig G M x \<pi> \<sigma>))"
\end{IsabelleSnippet}
\end{figure*}

%% file: removevtxcardmatching_listing.tex
\begin{figure*}[t]
\begin{IsabelleSnippet}[label=isa:removevtxcardmatching]{Formal statement of Lemma~\ref{lem:removevtxcardmatching}.}
lemma ranking_matching_card_leq_on_perfect_matching_graph:
  assumes "ranking_matching G M \<pi> \<sigma>" "ranking_matching (make_perfect_matching G N) M' \<pi> \<sigma>"
  shows "card M' \<le> card M"
\end{IsabelleSnippet}
\end{figure*}

%% file: rankcorrectgeneral_listing.tex
\begin{figure*}[h]
\begin{IsabelleSnippet}[label=isa:rankcorrectgeneral]{The formalisation of Theorem~\ref{thm:rankcorrectgeneral}.}
lemma comp_ratio_no_limit: 
  "measure_pmf.expectation ranking_prob card / (card V) \<ge> 1 - (1 - 1/(card V + 1)) ^ (card V)" 
\end{IsabelleSnippet}
\end{figure*}